\documentclass[a4paper,twocolumn]{quantumarticle}
\pdfoutput=1
\usepackage[utf8]{inputenc}
\usepackage[english]{babel}
\usepackage[T1]{fontenc}
\usepackage{amsmath}
\usepackage{hyperref}
\usepackage[numbers]{natbib}
\usepackage{amssymb}
\usepackage{amsthm}
\usepackage{xcolor}
\usepackage{bm} % for bold math symbols
\usepackage{wasysym} % for \checked symbol

\newtheorem{definition}{Definition}
\newtheorem{lemma}{Lemma}
\newtheorem{proposition}{Proposition}
\newtheorem{algorithm}{Algorithm}

\definecolor{darkcyan}{rgb}{0.0, 0.55, 0.55}
\newcommand{\todo}[1]{{\color{red}{#1}}}

 \newcommand{\halfquad}{\hspace{0.5em}} 
\begin{document}

\title{Possibilistic approach to network nonlocality}

\author{Antoine Restivo}
\affiliation{Department of Applied Physics University of Geneva, 1211 Geneva, Switzerland}
\affiliation{Centre for Quantum Information and Communication, École polytechnique de Bruxelles, Université libre de Bruxelles, CP 165, 1050 Brussels, Belgium}
\author{Nicolas Brunner}
\affiliation{Department of Applied Physics University of Geneva, 1211 Geneva, Switzerland}
\author{Denis Rosset}
\affiliation{Department of Applied Physics University of Geneva, 1211 Geneva, Switzerland}

\maketitle

\begin{abstract}
The investigation of Bell nonlocality traditionally relies on joint probabilities of observing certain measurement outcomes. In this work we explore a possibilistic approach, where only patterns of possible outcomes matter, and apply it to Bell nonlocality in networks with independent sources. We present various algorithms for determining whether a given outcome pattern can be achieved via classical resources or via non-signaling resources. Next we illustrate these methods considering the triangle and square networks (with binary outputs and no inputs), identifying patterns that are incompatible with the network structure, as well as patterns that imply nonlocality. In particular, we obtain an example of quantum nonlocality in the square network with binary outcomes. Moreover, we show how to construct certificates for detecting the nonlocality of a certain pattern, in the form of nonlinear Bell-type inequalities involving joint probabilities. Finally, we show that these inequalities remain valid in the case where the sources in the network become partially correlated. 
\end{abstract}

\section{Introduction}

Recently, growing interest has been devoted to the question of quantum
nonlocality in networks. From a fundamental perspective, this new area of
research aims at characterizing quantum correlations in a network
configuration, where several independent sources distribute entangled states
to subsets of nodes (parties). From a more practical point of view, these
ideas are also connected to the development of future quantum
communication networks.

This area of research already brought significant contributions; see e.g. \cite{Tavakoli_Review} for a recent review. The key idea is to characterize correlations in networks under the assumption that the different sources are independent from each other \cite{Branciard_2010}. Notably, it
was shown that the network configuration allows for novel forms of quantum
nonlocal correlations, departing radically from standard quantum Bell
nonlocality. The first striking example is the fact that quantum nonlocality
can be demonstrated without the needs of different measurement settings (i.e.
each party performs a single fixed measurement) \cite{Fritz_2012,Branciard_2012}. This effect is known as
``quantum nonlocality without inputs'', of which more examples have been reported \cite{fraser2018causal,Renou_2019,Renou2022,Gisin_2019,Abiuso2022,Boreiri2022}. Moreover, the use of non-classical
measurements (such as the well-known Bell-state measurement) allows for
quantum nonlocal correlations that are genuine to the network structure \cite{Supic2022,supic2022b}.

More generally, it is fair to say that a general
understanding of nonlocality in networks is still far away. Among the main
challenges is the characterisation of local correlations in a given network
structure (i.e. characterize those probability distributions that admit a
classical model), which turns out to be extremely challenging. All methods
developed in the context of (standard) Bell nonlocality are mostly useless
here, and novel methods must be developed, see e.g. \cite{Branciard_2010,Chaves2015, Rosset2016Bell,Wolfe2019,navascues2017inflation,Weilenmann_2018,Aberg2020,wolfe2019quantum,krivachy_neural_2020}. The main difficulty arises from the independence assumption of the various sources in the network, which makes the problem nonlinear and highly
nontrivial. Another challenging problem is to characterize the limits on
correlations imposed by the no-signaling principle \cite{Henson_2014,Wolfe2019,gisin2020constraints}. Both of these questions
are fundamental towards a better understanding of quantum nonlocality in
networks.

In the present work, we explore a possibilistic approach to the question of
nonlocality in networks. This approach provided deep insight in the
context of standard Bell nonlocality (most notably through the well-known
Hardy \cite{Hardy} and Greenberger-Horne-Zeilinger \cite{GHZ,Mermin} paradoxes). The possibilistic approach in the context of networks has been proposed by Wolfe, Spekkens and Fritz~\cite{Wolfe2019}, and Fraser~\cite{Fraser2020} developed an efficient algorithm for the case of local models, but no systematic investigation has been reported so far. As we will
see, this offers a new perspective on the problem of characterizing local,
quantum and no-signaling correlations in networks. 

The main idea is to move
away from probabilities, and consider only possibilities. That is, we are only
interested to know if an event (e.g. a certain combination of measurement
outputs) occurs with some non-zero probability or, on the contrary, is
impossible (i.e. has probability zero). We then ask which patterns (i.e. a
combination of possible and impossible events) are compatible with a local,
quantum and non-signaling model for a given network structure. We present
several methods to attack this question, relying on the
inflation technique, SAT solvers and efficient combinatorial algorithms. Next
we apply it to specific examples, namely the well-known ``triangle network''
and the square network (both for the case of binary outputs and without
inputs) and classify all patterns. Moreover, our methods allow for the
derivation of non-linear Bell-type inequalities for networks, which can then
be applied to actual probability distributions. We use these inequalities to
show that any pattern that is incompatible with a local model (assuming fully
independent sources) is also incompatible with a local model where the
different sources can be partially correlated \cite{Supic}.

We start in Section~\ref{Sec:GeneralizedBell} by describing the generalized
Bell scenarios in which we operate, and define not only the probability distributions
that are observed, but also the possibilistic patterns and their symmetries. We mention the triangle and square scenarios,
whose orbits we classify
in this paper. We describe the different sets of correlations (local, quantum and
non-signaling) and list the algorithms that can be used to test for membership.
In Section~\ref{Sec:NonsignalingCorrelations} we describe the algorithms
to decide whether a pattern or probability distribution is non-signaling. 
We define formally the inflations of a given scenario and the possibilistic
approach to the inflation technique. In Section~\ref{Sec:LocalCorrelations}
we do the same for the set of network-local correlations. We then
present in Sections~\ref{Sec:TriangleScenario} and~\ref{Sec:SquareScenario}
the classification results for the triangle and square scenarios, respectively. We also present an instance of quantum nonlocality in the square network with binary outcomes. Finally, in Section~\ref{Sec:Relaxation}, we show that we can derive inequalities
robust under a relaxation of the source independence assumption.

\section{Generalized Bell-like scenarios}
\label{Sec:GeneralizedBell}

We consider a generalized Bell-like scenario containing sources and observers,
where the sources distribute information of classical, quantum
or post-quantum nature to the observers. In turn, each observer processes
this information to provide a classical outcome taken from a finite set of
values; for simplicity we restrict observers to a single fixed
measurement. This leads to the formal definition below.

\begin{definition}
  A {\em scenario } $\mathbb{S}$ is described by a directed bipartite graph
  $\mathbb{S} = \left( \bm{\mathcal{S}}, \bm{\mathcal{O}},
  \mathcal{E} \right)$, where a set of sources $\mathcal{S}= \{ \mathcal{S}_i
  \}_{i \in \mathcal{I}}$ and a set of observers $\mathcal{O}= \{
  \mathcal{O}_j \}_{j \in \mathcal{J}}$ are connected by directed edges
  $\mathcal{E} \subseteq \{ (\mathcal{S}_i, \mathcal{O}_j) : i \in
  \mathcal{I}, j \in \mathcal{J} \}$, with the index sets $\mathcal{I}= \{ 1,
  \ldots, | \mathcal{I} | \}$ and $\mathcal{J}= \{ 1, \ldots, | \mathcal{J} |
  \}$. Each observer $\mathcal{O}_j$ produces an outcome $o_j \in \{ 1,
  \ldots, n_j \}$, and we include in $\mathbb{S}$ the number of outcome values
  $\{ n_j \}$ as a label on the observer vertices $\bm{\mathcal{O}}$.
\end{definition}

We observe the correlations in this scenario and describe them using a
probability distribution $P_{\{ \mathcal{O}_j \}_{j \in \mathcal{J}}} (o_1,
\ldots, o_{| \mathcal{J} |})$. The coefficients in $P$ are
nonnegative and obey the normalization constraint:
\begin{equation}
  \label{Eq:Normalization} \sum_{o_1, \ldots, o_{| \mathcal{J} |}} P (o_1,
  \ldots, o_{| \mathcal{J} |}) = 1.
\end{equation}
Without loss of generality, we assume that each observer is connected to at
least one source, otherwise the probability distribution can be factorized to
leave those observers out of the questions we tackle.

Geometrically speaking, we can enumerate the coefficients of $P_{\{
\mathcal{O}_j \}_{j \in \mathcal{J}}}$ into a vector $\vec{P}_{\mathcal{J}}
\in \mathbb{R}^{n_1 \ldots n_{| \mathcal{J} |}}$; in this paper we chose to
increment the outcome of the last observer first. Accordingly, we enumerate
the vectors of the Euclidean basis as $\{ \vec{e}_{o_1 \ldots o_{| \mathcal{J}
|}} \}_{o_1 \ldots o_{| \mathcal{J} |}}$.

\begin{definition}
  The elements of the Euclidean basis of the vector space of probability
  distributions are written:
  \begin{equation}
    \label{Eq:Euclidean} \left[ o_1 {\ldots o_{| \mathcal{J} |}}  \right]
    \equiv \vec{e}_{o_1 \ldots o_{| \mathcal{J} |}}
  \end{equation}
\end{definition}

We thus have:
\begin{equation}
  \label{Eq:ProbExpand} \vec{P} = \sum_{o_1 \ldots o_{| \mathcal{J} |}} P (o_1
  \ldots o_{| \mathcal{J} |}) \cdot [o_1 \ldots o_{| \mathcal{J} |}]
\end{equation}
\begin{definition}
  \label{Def:Marginal}Let $\mathcal{K} \subseteq \mathcal{J}$ be indices that
  describe a subset of observers. We write $P_{\{ \mathcal{O}_j \}_{j \in
  \mathcal{K}}}$ the corresponding {\em marginal probability distribution}
  and likewise the probability vector $\vec{P}_{\mathcal{K}}$.
\end{definition}

This vector can be computed according to
\begin{equation}
  \label{Eq:Marginal} \vec{P}_{\mathcal{K}} = M_{\mathcal{K}, \mathcal{J}}
  \cdot \vec{P}_{\mathcal{J}}
\end{equation}
where $M_{\mathcal{K}, \mathcal{J}}$ is a matrix with $0 / 1$ entries
that represents the marginalization step.
Probability distributions are described with nonnegative real
coefficients; marginalization and most operations used in this paper only
involve addition and multiplication. Thus the relevant algebraic structure is
the commutative semiring $(\mathbb{R}^{\geqslant 0}, +, \cdot)$ where
$\mathbb{R}^{\geqslant 0} = [0, \infty [$.

\subsection{Possibilities and patterns}\label{Sec:PossibilitiesAndPatterns}

We can obtain good insights about a given scenario by considering
{\em possibilities} instead of probabilities.

\begin{definition}
  A {\em possibility} $\mathsf{p} \in \{ \oslash, \checked \} = \mathbb{P}$
  is a binary value where $\oslash$ designates that a given
  event will never occur, while $\checked$
  designates that it can happen with non-zero probability.
\end{definition}

Moving to a possibilistic picture is of great interest since
for a given number of parties, the associated set of patterns is finite.
Indeed, for $n$ observers with $m$ possible outcomes, one can define $m^n$
possible patterns. Hence, it is possible to sort them completely. As we will
see, one can directly map the operations on possibilities to binary Boolean
algebra. This peculiarity allows us to formulate the marginal compatibility
problem into a Boolean satisfiability problem, which is easier to solve than
linear programs.

\

For two disjoint events with possibilities $\mathsf{p}_1$ and
$\mathsf{p}_2$, we can compute the possibility of either or both of them
happening as $\mathsf{p}_1 + \mathsf{p}_2$ with
\begin{equation}
  \oslash + \oslash = \oslash, \qquad \oslash + \checked = \checked + \oslash
  = \checked + \checked = \checked .
\end{equation}
Considering now two uncorrelated events with possibilities
$\mathsf{p}_1$ and $\mathsf{p}_2$, the possibility for both to occur
$\mathsf{p}_1 \cdot \mathsf{p}_2$ is given by
\begin{equation}
  \oslash \cdot \oslash = \oslash \cdot \checked = \checked \cdot \oslash =
  \oslash, \qquad \checked \cdot \checked = \checked .
\end{equation}
The set $\mathbb{P} = \{ \oslash, \checked \}$, equipped with addition $+$ and
the multiplication $\cdot$ described above is a commutative semiring
$(\mathbb{P}, +, \cdot)$. The semiring is ordered ($\oslash < \checked$), and
this order is compatible with addition.

Given a probability $P \in \mathbb{R}^{\geqslant 0}$, the corresponding
possibility $\mathsf{p}$ is
\begin{equation}
  \mathsf{p} = \left\{\begin{array}{ll}
    \oslash, & P = 0,\\
    \checked, & P > 0.
  \end{array}\right.
\end{equation}
The latter transformation defines a morphism $\varphi : \mathbb{R}^{\geqslant
0} \rightarrow \mathbb{P}$, and it is easy to verify that this morphism
preserves the semiring structure ($\varphi (x + y) = \varphi (x) + \varphi
(y)$, $\varphi (x \cdot y) = \varphi (x) \cdot \varphi (y)$) and the ordering
($x \leqslant y \Rightarrow \varphi (x) \leqslant \varphi (y)$). Using this
morphism, we can translate statements about probabilities into statements
about possibilities.

Two particular transformations are interesting. Let $x, y, z \in \mathbb{P}$.
Then:
\begin{equation}
  \label{Eq:SumToNope} x + y + z = \oslash \quad \Rightarrow \quad x = y = z =
  \oslash,
\end{equation}
\begin{equation}
  \label{Eq:SumToOK} x + y + z = \checked \halfquad \Rightarrow \halfquad (x =
  \checked) \vee (y = \checked) \vee (z = \checked) .
\end{equation}
From the probability distribution $\vec{P}_{\mathcal{J}} \in \mathbb{R}^N$, we
compute the {\em pattern} $\vec{\mathsf{P}}_{\mathcal{J}} \in
\mathbb{P}^N$ using the morphism $\vec{\mathsf{P}}_{\mathcal{J}} =
\varphi (\vec{P}_{\mathcal{J}})$. Each pattern
$\vec{\mathsf{P}}_{\mathcal{J}}$ corresponds to many distributions
$\vec{P}_{\mathcal{J}}$, and we call the set of those preimages
\begin{equation}
  \label{Eq:Realizations} \varphi^{- 1}
  (\vec{\mathsf{P}}_{\mathcal{J}}) \subset \mathbb{R}^N
\end{equation}
the {\em{realizations}} of the pattern
$\vec{\mathsf{P}}_{\mathcal{J}}$.

We note that not all patterns are physical. Using the morphism $\varphi$, we
rewrite \eqref{Eq:Normalization} into a possibilistic statement
\begin{equation}
  \label{Eq:PossibiliticNormalization} \sum_{o_1, \ldots, o_{| \mathcal{J} |}}
  \mathsf{P} (o_1, \ldots, o_{| \mathcal{J} |}) = \checked
\end{equation}
which is not satisfied by the pattern
$\vec{\mathsf{P}}_{\mathcal{J}} = (\oslash, \oslash, \ldots,
\oslash)$.

In a given scenario, the number of patterns satisfying the normalization
condition is $2^{o_1 \ldots o_{| \mathcal{J} |}} - 1$.

\subsection{Symmetries}

We work within the device-independent framework, and thus the labeling of
sources, observers and outcomes has no specific meaning. We consider in turn
relabelings of outcomes and relabelings of sources/observers.

We first consider the relabeling of outcomes. For a given $j$, consider the
classical outcome $o_j \in \{ 1, \ldots, | n_j | \}$. We can permute those
labels with a permutation $\pi_j \in S_{n_j}$ without changing the underlying
physics, with $S_{n_j}$ is the symmetric group of degree $n_j$.

\begin{definition}
  The {\em{group of outcome relabelings}} $G_{\textup{outcomes}}$ is the
  direct product $G_{\textup{outcomes}} = S_{n_1} \times \ldots \times S_{n_{|
  \mathcal{J} |}}$. Its elements are written $\pi = (\pi_1, \ldots, \pi_{|
  \mathcal{J} |}) \in G_{\textup{outcomes}}$.
\end{definition}

Consider $\pi \in G_{\textup{outcomes}}$ and a probability distribution
$\vec{P}$, we write using~\eqref{Eq:ProbExpand}:
\begin{equation}
  \label{Eq:RelabelingOutcomes} 
  \begin{split}
  \pi (\vec{P}) = \sum_{o_1 \ldots o_{|
  \mathcal{J} |}} P (o_1 \ldots o_{| \mathcal{J} |}) \cdot 
  \\
  \cdot [\pi_1 (o_1), \ldots, \pi_{| \mathcal{J} |} (o_{| \mathcal{J} |})] .
  \end{split}
\end{equation}
We now consider relabelings of sources and observers.

\begin{definition}
  The {\em{group of observers}} relabelings $G_{\textup{observers}}$ is
  derived from the automorphism group of the directed bipartite graph
  $\mathbb{S}$, by considering its action on observers: $G_{\textup{obsevers}}
  \subset S_{| \mathcal{J} |}$.
\end{definition}

We write the action of $\sigma \in G_{\textup{observers}}$ on a probability
distribution:
\begin{equation}
  \label{Eq:RelabelingObservers} \sigma (\vec{P}) = \!\! \sum_{o_1 \ldots o_{|
  \mathcal{J} |}} \!\!\!\! P (o_1 \ldots o_{| \mathcal{J} |}) \cdot [o_{\sigma^{- 1}
  (1)} \ldots o_{\sigma^{- 1} (| \mathcal{J} |)}] .
\end{equation}
\begin{definition}
  A relabeling $g = (\sigma, \pi)$ in the scenario $\mathbb{S}$ is
  composed of a member of $G_{\textup{observers}}$ and $G_{\textup{outcomes}}$.
  The {\em{relabeling group}} $G = G_{\textup{observers}} \ltimes
  G_{\textup{outcomes}}$ is a semidirect product constructed using the natural
  action of $G_{\textup{observers}}$ on the elements of $G_{\textup{outcomes}}$.
\end{definition}

A relabeling $g \in G$ acts on probability distributions as $g (\vec{P}) =
(\sigma, \pi) (\vec{P}) = \sigma (\pi (\vec{P}))$. However, all we care about
in this paper is the enumeration of relabelings in $G$, and the exact
structure of $G$ is not relevant for that purpose.

Relabelings of probability distributions translate into relabelings of
patterns. The group $G$ is especially interesting when studying all possible
patterns of a scenario, as to group them.

\begin{definition}
  Under the action of $G$, a pattern $\vec{\mathsf{P}}$ correspond
  to an {\em{orbit}} $\vec{\mathsf{P}}^G$ of patterns with similar
  physical properties:
  \begin{equation}
    \vec{\mathsf{P}}^G = \{ g (\vec{\mathsf{P}}) : g \in
    G \} .
  \end{equation}
\end{definition}

As an example, the nonphysical pattern $(\oslash, \oslash, \ldots, \oslash)$
is the only member of its orbit.

To make the definitions above concrete, we turn to the two scenarios
considered in this paper.

\begin{figure}[t]
  \centering
  \includegraphics{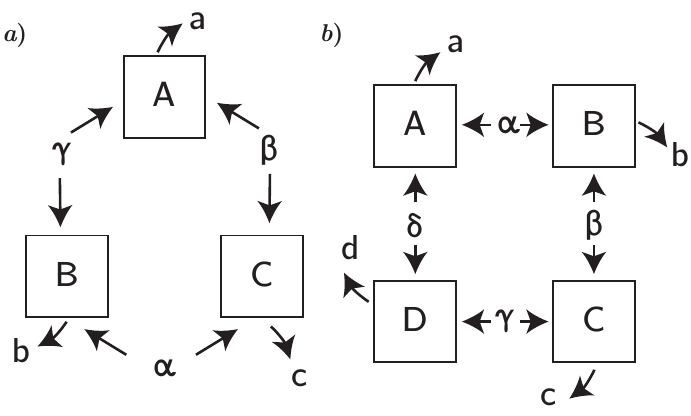}
  \caption{\label{Fig:Scenarios}Scenarios considered in this paper: a) the
  triangle scenario, b) the square scenario}
\end{figure}

\subsection{The triangle scenario}
\label{Sec:TriangleScenarioDef}

The triangle scenario involves three sources $\mathcal{S}= \{ \alpha, \beta,
\gamma \}$ and three observers $\mathcal{O}= \{ A, B, C \}$, with connections
$\mathcal{E}= \{ (\alpha, B), (\alpha, C), (\beta, A), (\beta, C), (\gamma,
A), (\gamma, B) \}$, as shown in Figure~\ref{Fig:Scenarios}a. We consider
binary outcomes $a, b, c \in \{ 0, 1 \}$. A probability distribution
$\vec{P}_{A B C} \in \mathbb{R}^8$ in this scenario has the coefficients:
\begin{equation}
\begin{split}
  \vec{P}_{A B C} = (P (000), P (001), P (010), P (011), \ldots \\
  P(100), P(101), P(110), P (111))^{\top} .
  \end{split}
\end{equation}
The symmetry group $G_{\textup{outcomes}}$ has 8 elements, and the automorphism
group $G_{\textup{observers}}$ is the symmetry group of the triangle, i.e. the
dihedral group of order 6: $D_6 = S_3$.

This scenario contains $255 = 2^{2 \cdot 2 \cdot 2} - 1$ patterns satisfying
normalization. Using relabelings, they are grouped into $21$ orbits.

\subsection{The square scenario}
\label{Sec:SquareScenarioDef}

The square scenario involves four sources $\mathcal{S}= \{ \alpha, \beta,
\gamma, \delta \}$ and four observers $\mathcal{O}= \{ A, B, C, D \}$, with
connections
\begin{equation}
\begin{split}
\mathcal{E}= \{ (\alpha, A), (\alpha, B), (\beta, B), (\beta, C), \\
(\gamma, C), (\gamma, D), (\delta, D), (\delta, A) \},
\end{split}
\end{equation}
as shown in
Figure~\ref{Fig:Scenarios}b. Again we consider binary outcomes $a, b, c, d \in
\{ 0, 1 \}$. We write a probability distribution $\vec{P}_{A B C D} \in
\mathbb{R}^{16}$. The symmetry group $G_{\textup{outcomes}}$ has $16$ elements,
and the automorphism group $G_{\textup{observers}}$ is the symmetry group of the
square, i.e. the dihedral group of order $8$, namely $D_8$.

The scenario contains $65535 = 2^{2 \cdot 2 \cdot 2 \cdot 2} - 1$ patterns
satisfying normalization. Under relabelings, they are grouped into
$804$ orbits.

\subsection{Correlation sets and decision problems}

Depending on the type of resources that the sources $\{ \mathcal{S}_i \}$
distribute, we may observe different sets of correlations. Among the sets of
interest, we distinguish three.

\begin{definition}
  The {\em{nonsignaling set}} $\mathcal{N} \subset \mathbb{R}^N$, containing
  probability distributions that do not enable signaling.
\end{definition}

\begin{definition}
  The {\em{quantum set}} $\mathcal{Q} \subseteq \mathcal{N} \subset
  \mathbb{R}^N$, containing probability distributions obtained when the
  sources distribute quantum states which are measured by the observers
  according to the axioms of quantum mechanics.
\end{definition}

\begin{definition}
  The {\em{local set}} $\mathcal{L} \subseteq \mathcal{Q} \subseteq
  \mathcal{N} \subset \mathbb{R}^N$, containing probability distribution
  obtained by processing classical information. There we identify sources with
  local hidden variables, and the processing follows from classical
  probability axioms.
\end{definition}

For each of these sets, we aim to answer the decision problem: is $\vec{P} \in
\mathcal{N}, \mathcal{Q}, \mathcal{L}$? When the answer is, e.g., $\vec{P} \in
\mathcal{L}$, we may want a {\em{model}} that shows explicitly how, e.g.,
$\vec{P}$ can be obtained using local hidden variables. When the answer is
negative (e.g. $\vec{P} \notin \mathcal{L}$), we may want a
{\em{certificate}}, i.e. a statement of non-membership that is nevertheless
easier to verify than solving the decision problem once again.

As the labeling of sources, observers and outcomes is arbitrary, for a given
relabeling $g \in G$, we have
\begin{equation}
  g (\vec{P}) \in \mathcal{N} \quad \Leftrightarrow \quad \vec{P} \in
  \mathcal{N},
\end{equation}
and the same for the sets $\mathcal{Q}$ and $\mathcal{L}$.

As we discuss in the next sections, those decision problems are hard and scale
badly with the number of sources, observers and outcomes. Solving those
problems on patterns can be much easier. We write $\mathsf{N}$, $\mathsf{Q}$,
$\mathsf{L}$ for the sets of patterns corresponding to $\mathcal{N},
\mathcal{Q}, \mathcal{L}$: the pattern $\vec{\mathsf{P}} \in \mathsf{N}$ if
there exists $\vec{P} \in \mathcal{N}$, and so on.

As the number of all patterns for a given scenario is finite, we can group
them according to the layers shown in Figure~\ref{Fig:PatternSets}. The
techniques and algorithms are described in Table~\ref{Tab:DecisionProblems}.

\begin{figure}[t]
  \centering
  \includegraphics{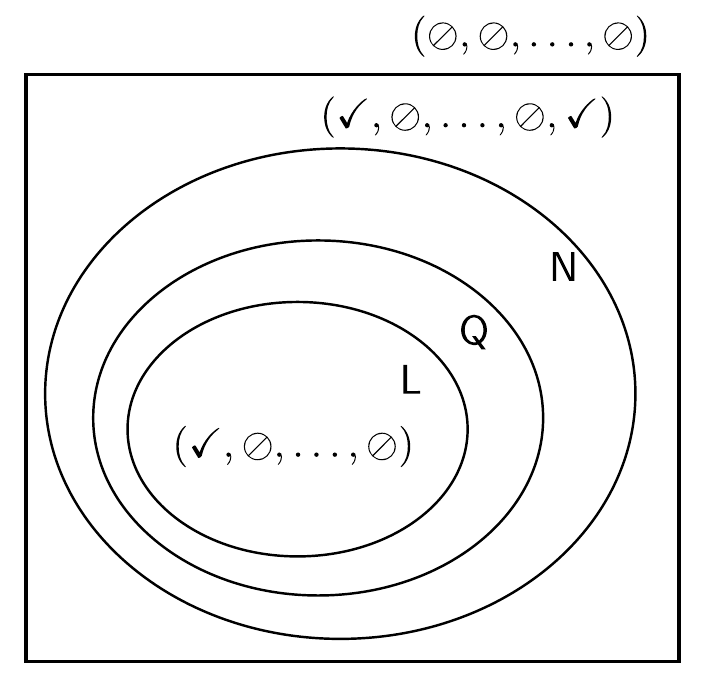}
  \caption{\label{Fig:PatternSets}Local ($\mathsf{L}$), quantum ($\mathsf{Q}$)
  and nonsignaling ($\mathsf{N}$) sets of patterns, with examples of patterns
  in the triangle scenario. The rectangle contains only patterns that satisfy
  the normalization condition. In the triangle scenario, the pattern $[000]$
  is obviously local, and the pattern $[000] + [111] \notin \mathsf{N}$.}
\end{figure}

\begin{table*}
  \begin{tabular}{|l|l|l|l|l|}
      \hline
      Decision problem & Exact algorithm & Relaxation: find a model &
      Relaxation: find imp. certificate\\
      \hline
      $\vec{P}_{\textup{test}} \in \mathcal{N}$ & ? & ? &
      Alg.~\ref{Alg:NSProbInflation}\\
      \hline
      $\vec{\mathsf{P}}_{\textup{test}} \in \mathsf{N}$ & ? & ? &
      Alg.~\ref{Alg:NSPatInflationSat} and~\ref{Alg:NSPatInflationBitVector}\\
      \hline
      $\vec{P}_{\textup{test}} \in \mathcal{Q}$ & ? & Nonconvex
      optimization~{\cite{Liang2007a}} & Inflation\&SDP {\cite{Wolfe2021a,Ligthart2021}}\\
      \hline
      $\vec{\mathsf{P}}_{\textup{test}} \in \mathsf{Q}$ & ? & ? & ?\\
      \hline
      $\vec{P}_{\textup{test}} \in \mathcal{L}$ & Alg.~\ref{Alg:LProbModel} &
      Same as left, red. cardinality & Alg.~\ref{Alg:LProbInflation}\\
      \hline
      $\vec{\mathsf{P}}_{\textup{test}} \in \mathsf{L}$ &
      Alg.~\ref{Alg:LPatModel}, Alg.~\ref{Alg:LocalTC} & Same as left, red. cardinality &
      Alg.~\ref{Alg:LPatInflation}\\
      \hline
  \end{tabular}
  \caption{
  \label{Tab:DecisionProblems}Decision problems and available
  techniques 
  }
\end{table*}

\subsection{Certificates}

We now turn to certificates. They describe in a compact manner that a set of
distributions or patterns is {\em{not}} in a correlation set. In our paper,
they are the result of an algorithm solving a decision problem, and
subsequently testing the validity of a certificate is easier than solving the
decision problem in the first place.

\begin{definition}
  \label{Def:Valuation}Following~{\cite{Wolfe2019}}, a {\em{valuation}} $V$
  is an assignment of outcomes to a subset of parties of a scenario.
\end{definition}

For example, in the triangle scenario $\{ a = 0, c = 1 \}$ is a valuation on
the observers A and C. It corresponds to the probability $P (V) =
P_{\textup{AC}} (01)$.
For a given scenario $\mathbb{S}$, we consider the set $\bm{V}$ of all
valuations. Let $V_1, \ldots, V_m \in \bm{V}$ be arbitrary valuations.

\begin{definition}
  \label{Def:Monomial}A {\em{monomial}} $M = P (V_1) \ldots P (V_m)$ of
  degree $\deg M = m$ is defined by a sequence $(V_1, \ldots, V_m)$ of $m$
  valuations. For a given scenario $\mathbb{S}$, we write $\bm{M}$ the
  set of all its monomials.
\end{definition}

For example, the monomial $P_{\textup{ABC}} (000)$ has degree 1 while
$P_{\textup{A}} (0) P_{\textup{B}} (0)$ has degree 2.

\begin{definition}
  A {\em{certificate}} $I$ for the correlation set $\mathcal{X} \in \{
  \mathcal{L}, \mathcal{Q}, \mathcal{N} \}$ is a polynomial $I \in \mathbb{R}
  \bm{M}$ (a formal sum over $\bm{M}$) with the condition that
  \begin{equation}
    \label{Eq:Certificate} \vec{P} \in \mathcal{X} \quad \Rightarrow \quad I
    (\vec{P}) \geqslant 0.
  \end{equation}
\end{definition}

The degree of a polynomial $I \in \mathbb{R} \bm{M}$ is given by the
maximal degree of all its monomials.

\section{Nonsignaling correlations}
\label{Sec:NonsignalingCorrelations}
\begin{figure}[t]
  \centering
  \includegraphics{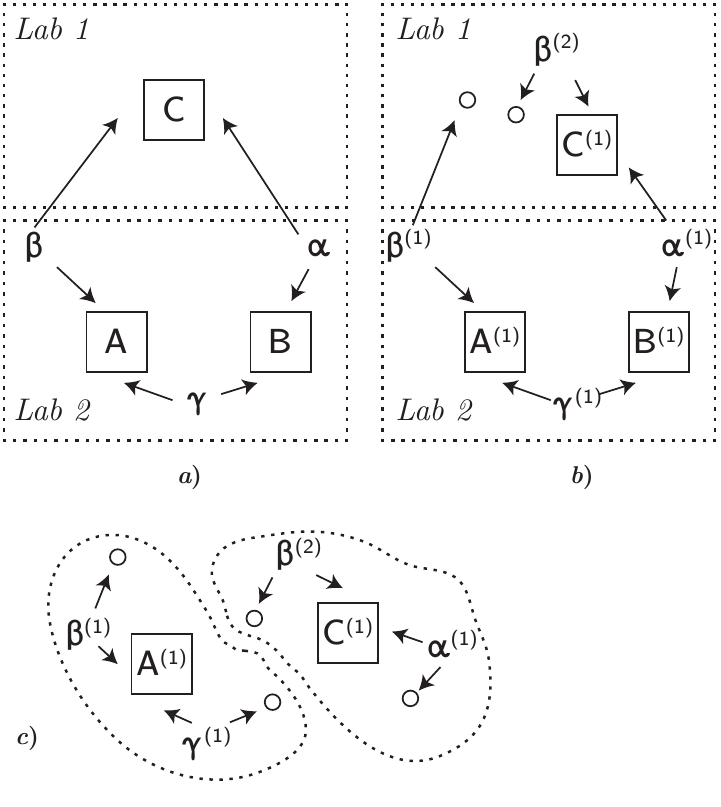}
  \caption{\label{Fig:Nonsignaling}Construction of nonsignaling condition: a)
  the original triangle scenario, b) a variant of the scenario involving a
  copy of the source $\gamma$, c) a variant of the scenario involving a copy
  of the source $\gamma$ and removing $B$.}
\end{figure}

In~\eqref{Eq:Normalization} and in~\eqref{Eq:PossibiliticNormalization}, we
asserted that correlations must obey the normalization constraint. However,
there are other constraints that come into play. We start with an example.

\subsection{Nonsignaling constraints}\label{Sec:PossibilisticNS}

Let us look at the triangle scenario in Figure~\ref{Fig:Nonsignaling}a. In
Figure~\ref{Fig:Nonsignaling}b, we construct a variant of the scenario by
duplicating the source $\beta$; we also added indices to sources and
observers to distinguish them from the original scenario.

In both scenarios, if we examine the probability distributions $\vec{P}_{A B}$
and $\vec{P}_{A^{(1)} B^{(1)}}$, these distributions must be identical due to
the no signaling principle, as the devices present in the second lab are
identical. Indeed, if $\vec{P}_{A B} \neq \vec{P}_{A^{(1)} B^{(1)}}$, then the
first lab could signal to the second lab by changing the wiring.
We thus have $\vec{P}_{A B} = \vec{P}_{A^{(1)} B^{(1)}}$. By changing the lab
boundaries, we also get $\vec{P}_{B C} = \vec{P}_{B^{(1)} C^{(1)}}$.
Finally, we observe that removing $B^{(1)}$ leads to the scenario in
Figure~\ref{Fig:Nonsignaling}c, where we must have $P_{A^{(1)} C^{(1)}} (a c)
= P_{A^{(1)}} (a) P_{C^{(1)}} (c)$.

We now consider the pattern $\vec{\mathsf{P}}_{A B C}^{\textup{GHZ}} = [000] +
[111] = (\checked, \oslash, \oslash, \oslash, \oslash, \oslash, \oslash,
\checked)^{\top}$, slightly abusing the notation~\eqref{Eq:Euclidean}. If that
pattern obeys the nonsignaling principle, there must be a pattern
$\vec{\mathsf{P}}_{A^{(1)} B^{(1)} C^{(1)}}$ such that
\begin{align}
  \vec{\mathsf{P}}_{A^{(1)} B^{(1)} } &= \vec{\mathsf{P}}_{A B}, \nonumber \\
  \vec{\mathsf{P}}_{B^{(1)} C^{(1)} } &= \vec{\mathsf{P}}_{B C}, \nonumber \\
  \mathsf{P}_{A^{(1)} C^{(1)}} (a c) &= \mathsf{P}_A (a) \mathsf{P}_C (c) .
\end{align}
We have thus:
\begin{equation}
  \label{Eq:Cut1} \vec{\mathsf{P}}_{A^{(1)} B^{(1)}} =
  \vec{\mathsf{P}}_{B^{(1)} C^{(1)} } = (\checked, \oslash, \oslash, \checked)
  .
\end{equation}
We also have $P_A (a) = P_C (c) = \checked$ and thus:
\begin{equation}
  \label{Eq:Cut2} \vec{\mathsf{P}}_{A^{(1)} C^{(1)} } = (\checked, \checked,
  \checked, \checked)
\end{equation}
Using the statements~\eqref{Eq:SumToNope} and~\eqref{Eq:SumToOK}, we have:
\begin{equation*}
  \label{Eq:GHZContradiction} \begin{array}{|l|l|l|l|l|}
    \hline
    a & b & c & \mathsf{P}_{A_1 B_1 C_1} \textup{ of } \eqref{Eq:Cut1} &
    \mathsf{P}_{A_1 B_1 C_1} \textup{ of } \eqref{Eq:Cut2}\\
    \hline
    0 & 0 & 0 & \checked & \\
    \hline
    0 & 0 & 1 & \oslash & \chi_1\\
    \hline
    0 & 1 & 0 & \oslash & \\
    \hline
    0 & 1 & 1 & \oslash & \chi_2\\
    \hline
    1 & 0 & 0 & \oslash & \omega_1\\
    \hline
    1 & 0 & 1 & \oslash & \\
    \hline
    1 & 1 & 0 & \oslash & \omega_2\\
    \hline
    1 & 1 & 1 & \checked & \\
    \hline
  \end{array} \quad 
\end{equation*}
where $(\chi_1 = \checked) \vee (\chi_2 = \checked)$ and
$(\omega_1 = \checked) \vee (\omega_2 = \checked)$,
which leads to a contradiction.

We say that the pattern $\vec{\mathsf{P}}_{A B
C}^{\textup{GHZ}} = [000] + [111]$ is {\em{signaling-enabling}}; if we had a
scenario with those correlations, the modification of this scenario as in
Figure~\ref{Fig:Nonsignaling}b would enable signaling. This statement means
that all realizations of $\vec{\mathsf{P}}_{A B C}^{\textup{GHZ}}$ in the sense
of~\eqref{Eq:Realizations}, written
\begin{equation}
  \label{Eq:PGHZ} \vec{P}_{A B C}^{\textup{GHZ}} = (x, 0, 0, 0, 0, 0, 0, 1 - x),
\end{equation}
are also signaling-enabling for $0 < x < 1$. Finally, we note that the orbit of
$\vec{\mathsf{P}}^{\textup{GHZ}}$ is also {\em{signaling-enabling}}:
\begin{equation}
  \label{Eq:PatGHZOrbit}
  \begin{split}
  (\vec{\mathsf{P}}^{\textup{GHZ}})^G = \{ [000] + [111],
  [001] + [110], \\
  [010] + [101], [011] + [100] \} .
  \end{split}
\end{equation}

\subsection{Inflations and nonsignaling inflations}

In the Figure~\ref{Fig:Nonsignaling} we used in the previous section, we
presented an example of a modification of the original scenario. Such a
modification is called an {\em{inflation}}, following the process proposed
by~{\cite{Wolfe2019}}.

\begin{definition}
  An {\em{inflation}} contains copies of the original sources and observers,
  wiring them according to the template of the original scenario. In the
  inflation $\mathbb{S}' = (\mathcal{S}', \mathcal{O}', \mathcal{E}')$, we
  write with $\mathcal{S}' = \{ \mathcal{S}^{(k)}_i \}$ and $\mathcal{O}' = \{
  \mathcal{O}_j^{ (l)} \}$.
  We wire them according to the following rules.
  \begin{itemize}
  \item Compatible types:
  \begin{equation}
    \label{Eq:CompatibleTypes}(\mathcal{S}_i^{(k)},
    \mathcal{O}_j^{(l)}) \in \mathcal{E}' \Rightarrow (\mathcal{S}_i,
    \mathcal{O}_j) \in \mathcal{E},
  \end{equation}
  \item Parties receive a complete set of
    sources:
  \begin{equation}
    \label{Eq:CompleteSources} (\mathcal{S}_i, \mathcal{O}_j) \in \mathcal{E} \Rightarrow
    \forall l, \exists k, (\mathcal{S}_i^{(k)}, \mathcal{O}_j^{(l)}) \in
    \mathcal{E}',
  \end{equation}
  \end{itemize}
\end{definition}

Additionally, if the inflation is {\em{nonsignaling}}, we impose the
additional rule:
\begin{itemize}
    \item No duplication of sources:
\begin{equation}
  \label{Eq:NoDuplication} \{
  (\mathcal{S}_i^{(k)}, \mathcal{O}_j^{(l)}), (\mathcal{S}_i^{(k)},
  \mathcal{O}_j^{(l')}) \} \subset \mathcal{E}' \Rightarrow l = l' .
\end{equation}
\end{itemize}
Such inflations are referred to as non-fanout \cite{Wolfe2019}. In the triangle scenario, all the inflations have the shape of a ring, as
outlined in Figure~\ref{Fig:TriangleRings}. The complete family of
nonsignaling constraints is thus a sequence of rings containing $3 n$
observers for $n = 2, 3, \ldots$. The same reasoning applies to the square
scenario, where the complete family of nonsignaling constraints is expressed
by a sequence of rings of length $4 n$ for $n = 2, 3, \ldots$.

\begin{figure}[t]
  \centering
  \includegraphics{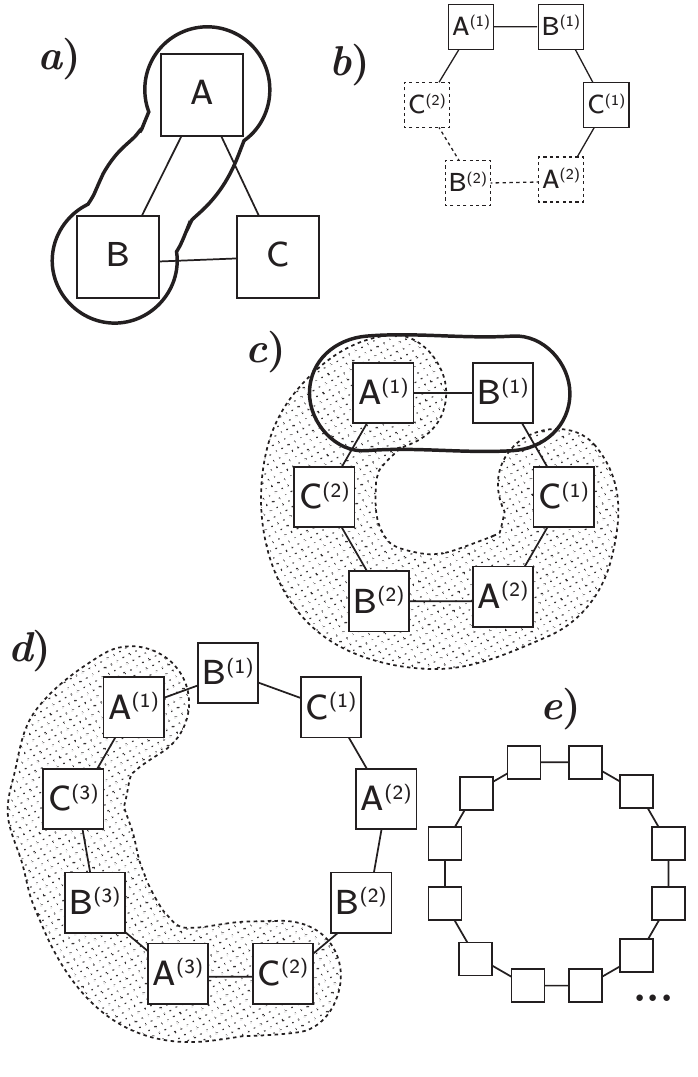}
  \caption{\label{Fig:TriangleRings}Complete family of inflations for the
  triangle scenario, where the sources vertices are omitted for readability.
  a) the original scenario, b) the $6$-ring contains the inflation of
  Figure~\ref{Fig:Nonsignaling}b as a subgraph, c) the 6-ring with subgraphs
  highlighted, in black a subgraph compatible with the original scenario, d)
  the 9-ring with a highlighted subgraph matching a subgraph of the 6-ring but
  not of the original scenario, e) a sketch of the 12-ring.}
\end{figure}

\subsection{AI-expressible subsets}

We consider subsets $\mathcal{K}' \subseteq \mathcal{J}'$ of observers in an
inflation $\mathbb{S}'$. When $\mathcal{K}' = \{ \mathcal{O}_j^{(\ell)} \}$ is
a singleton, the marginal probability distribution of the inflation scenario
matches the marginal probability distribution $\vec{P}$ of the original
scenario:
\begin{equation}
  \vec{P}^{\textup{inflation}}_{\mathcal{O}^{(\ell)}_j} =
  \vec{P}_{\mathcal{O}_j},
\end{equation}
This works as well when $\mathcal{K}' = \{ \mathcal{O}_j^{ (\ell)},
\mathcal{O}_{j'}^{ (\ell')} \}$ and one of the following holds true.
\begin{itemize}
  \item There are sources connected to both $\mathcal{O}_j^{(\ell)}$ and
  $\mathcal{O}_{j'}^{(\ell')}$, but $j \neq j'$. We consider the subgraph of
  $\mathbb{S}'$ after removal of the observers in $\mathcal{J}'
  \backslash\mathcal{K}'$ and of orphan sources (sources not connected to any
  observer). This subgraph matches the subgraph of the original scenario
  $\mathbb{S}$ after removal of the observers in $\mathcal{J}\backslash \{
  \mathcal{O}_j, \mathcal{O}_{j'} \}$ and orphan sources. Thus
  \begin{equation}
    \vec{P}^{\textup{inflation}}_{\mathcal{O}_j^{ (\ell)}
    \mathcal{O}_{j'}^{(\ell')}} (x y) = \vec{P}_{\mathcal{O}_j
    \mathcal{O}_{j'}} (x y),
  \end{equation}
  where $x, y$ are enumerated over the possible outcomes. This corresponds to
  the cut inflation in Figure~\ref{Fig:Nonsignaling}b, where
  $\vec{P}_{\mathcal{A}^{(1)} \mathcal{B}^{(1)}} (a b) =
  \vec{P}_{\mathcal{A}\mathcal{B}} (a b) .$
  
  \item There is no source $\mathcal{S}_i^k$ connected to both
  $\mathcal{O}_j^{(\ell)}$ and $\mathcal{O}_{j'}^{(\ell')}$. Then
  \begin{align}
    \vec{P}^{\textup{inflation}}_{\mathcal{O}_j^{ (\ell)}
    \mathcal{O}_{j'}^{(\ell')}} (x y) & =
    \vec{P}^{\textup{inflation}}_{\mathcal{O}_j^{ (\ell)}} (x)
    \vec{P}^{\textup{inflation}}_{\mathcal{O}_{j'}^{(\ell')}} (y) \nonumber \\
    & =
    \vec{P}_{\mathcal{O}_j} (x) \vec{P}_{\mathcal{O}_{j'}} (y),
    \end{align}
  again over all possible $x, y$. This corresponds to
  Figure~\ref{Fig:Nonsignaling}c where $\vec{P}_{\mathcal{A}^{(1)}
  \mathcal{C}^{(1)}} (a c) = \vec{P}_{\mathcal{A}} (a) \vec{P}_{\mathcal{C}}
  (c)$
\end{itemize}
Such sets $\mathcal{K}'$ are called {\em{ai-expressible}},
following~{\cite{Wolfe2019}}. The definition generalizes to any cardinality of
$\mathcal{K}'$.

\begin{definition}
  Given an inflation $\mathbb{S}'$ of a scenario $\mathbb{S}$, a subset
  $\mathcal{K}' \subseteq \mathcal{J}'$ of the observers is
  {\em{AI-expressible}} when the following holds: after removing the
  observers in $\mathcal{J}' \backslash\mathcal{K}'$ and orphan sources, the
  connected components of the resulting graph each match a subgraph of
  $\mathbb{S}$.
\end{definition}

For an AI-expressible set $\mathcal{K}'$, we $\mathcal{K}' =\mathcal{H}'_1
\cup \mathcal{H}'_2 \cup \cdots \cup \mathcal{H}'_{\left|
\bm{\mathcal{H}} \right|}$ as a disjoint union of observers for each of
its connected components. For each $\mathcal{H}'_h$, we write $\mathcal{H}_h$
the corresponding observers in the original scenario, after dropping the copy
index. Then we identify:
\begin{align}
  \label{Eq:Identify} 
  P_{\mathcal{K}'} \left( \bm{k} \right) &=
  P_{\mathcal{H}'_1 \ldots \mathcal{H}'_{\left| \bm{\mathcal{H}}
  \right|}} \left( \bm{h}_1 \ldots \bm{h}_{\left|
  \bm{\mathcal{H}} \right|} \right) \nonumber \\ 
  & = P_{\mathcal{H}'_1} \left(
  \bm{h}_1 \right) \ldots P_{\mathcal{H}'_{\left| \bm{\mathcal{H}}
  \right|}} \left( \bm{h}_{\left| \bm{\mathcal{H}} \right|}
  \right) \nonumber \\
  & = P_{\mathcal{H}_1} \left( \bm{h}_1 \right) \ldots
  P_{\mathcal{H}_{\left| \bm{\mathcal{H}} \right|}} \left(
  \bm{h}_{\left| \bm{\mathcal{H}} \right|} \right),
\end{align}
and remark that the r.h.s. of~\eqref{Eq:Identify} is a member of the monomials
$\bm{M}$ according to Definition~\ref{Def:Monomial}.

\subsection{Algorithms to compute with nonsignaling
constraints}\label{Sec:AlgorithmsNS}

We do not know of a technique that can test if a particular $\vec{P}_{A B C}$
is nonsignaling or signaling-enabling.

There is a relaxation that uses linear programming (LP), described
in~{\cite{Wolfe2019}}. It is a relaxation because we need to impose a cutoff
on the complexity of the inflation (in the triangle, that would be the ring
length); also because some nonlinear independence conditions cannot be
expressed using LP. Given a distribution $\vec{P}_{A B C}$ and a cutoff on the
complexity of the inflation, we can get two results: the distribution is
incompatible with this limited set of nonsignaling constraints, or the result
is inconclusive.

\begin{algorithm}
      \label{Alg:NSProbInflation}From {\cite{Wolfe2019}}. Decide on an
      inflation of the original scenario. Identify maximal AI-expressible sets
      of observers. Write all constraints between marginals of
      $\vec{P}_{\textup{inflation}}$ and the explicit values of the
      corresponding monomials in $\bm{M}$ in the form
      of~\eqref{Eq:Identify}. When the linear program is infeasible, we can
      extract a polynomial inequality $I \in \mathbb{R} \bm{M}$ in the
      form~\eqref{Eq:Certificate} that certifies that $\vec{P}_{\textup{test}}
      \notin \mathcal{N}$, along with a decomposition that allows the
      verification of the certificate without solving the linear program
      again.
\end{algorithm}

This relaxation quickly becomes expensive as the complexity of the inflation
grows. However, if we study patterns instead of probability distributions, the
problem simplifies. Indeed, the semiring $(\mathbb{P}, +, \cdot)$ is
isomorphic to the Boolean algebra $(\{ 0, 1 \}, \vee, \wedge)$: the
resulting Boolean satisfiability problem can be solved efficiently using SAT
solvers.

\begin{algorithm}
      \label{Alg:NSPatInflationSat}Similar as
      Algorithm~\ref{Alg:NSProbInflation}, but the equality constraints are
      written between marginals of $\vec{\mathsf{P}}_{\textup{inflation}} \in
      \mathbb{P}^{N'}$ and $\vec{\mathsf{P}}_{\textup{test}} \in \mathbb{P}^N$.
      It is also possible to incorporate independence conditions of the form
      $\mathsf{P}_{A^{(1)} C^{(1)}} (a_1 c_1) = \mathsf{P}_{A^{(1)}} (a_1)
      \mathsf{P}_{C^{(1)}} (c_1)$. The resulting problem is solved by a SAT
      solver. Some SAT solvers support the extraction of a {\em{proof}} or
      infeasibility certificate; but this is an active domain of
      research~{\cite{Baek2021}}.
\end{algorithm}

However, there is a simpler algorithm if we restrict ourselves to equality
constraints where the right-hand side is deduced directly from the pattern
under test.

\begin{algorithm}
      \label{Alg:NSPatInflationBitVector}There is an algorithm that solves a
      relaxation of the $\vec{\mathsf{P}}_{\textup{test}} \in
      \mathsf{N}$ decision problem in $\mathcal{O} (m \cdot N')$ time and
      $\mathcal{O} (N')$ space, if $m$ is the number of equality constraints
      and $N'$ the number of elements in $\vec{\mathsf{P}}_{\textup{inflation}}
      \in \mathbb{P}^{N'}$.
      
      The general principle is to identify all coefficients $= \oslash$ in
      $\vec{\mathsf{P}}_{\textup{inflation}}$ using~\eqref{Eq:SumToNope}, then
      look for a contradiction using a constraint $= \checked$
      using~\eqref{Eq:SumToOK}.
\end{algorithm}

We discuss Algorithm~\ref{Alg:NSPatInflationBitVector} using the example of
Section~\ref{Sec:PossibilisticNS}, where $\vec{\mathsf{P}}_{\textup{inflation}}
= \vec{\mathsf{P}}_{A^{(1)} B^{(1)} C^{(1)}}$ and
$\vec{\mathsf{P}}_{\textup{test}} = (\checked, \oslash, \oslash, \oslash,
\oslash, \oslash, \oslash, \checked)^{\top}$. We sort constraints according to
the right-hand side:
\begin{align}
  \mathsf{P}_{A^{(1)} B^{(1)}} (00) = \mathsf{P}_{A^{(1)} B^{(1)}} (11) &= \checked \nonumber \\
  \mathsf{P}_{B^{(1)} C^{(1)}} (00) = \mathsf{P}_{B^{(1)} C^{(1)}} (11) &= \checked \nonumber \\
  \mathsf{P}_{A^{(1)} C^{(1)}} (00) = \mathsf{P}_{A^{(1)} C^{(1)}} (01) &= \checked \nonumber \\
  \mathsf{P}_{A^{(1)} C^{(1)}} (10) = \mathsf{P}_{A^{(1)} C^{(1)}} (11) &= \checked 
\end{align}
and
\begin{align}
  \label{Eq:CutNopeRHS}
  \mathsf{P}_{A^{(1)} B^{(1)}} (01) = \mathsf{P}_{A^{(1)} B^{(1)}} (10) &= \oslash \nonumber \\
  \mathsf{P}_{B^{(1)} C^{(1)}} (01) = \mathsf{P}_{B^{(1)} C^{(1)}} (10) &= \oslash .
\end{align}
Let the hypothesis set $\mathbb{H} = \{ \oslash, ? \}$ where $\oslash$ means
the same as in $\mathbb{P}$ and $?$ is a placeholder for either $\oslash$ or
$\checked$. We initialize the hypothesis vector
$\vec{\mathsf{H}}_{\textup{inflation}} \in \mathbb{H}^{N'}$ as
$\vec{\mathsf{H}}_{\textup{inflation}} := (?, \ldots ., ?)^{\top}$. Then,
for each constraint with a right-hand side $= \oslash$, we identify the
elements of $\vec{\mathsf{H}}_{\textup{inflation}}$ that correspond to the sum
making up to the left-hand side. As $x + y + \ldots = \oslash$ implies $x = y
= \cdots = \oslash$, we set all these elements to $\oslash$.

For example, $\mathsf{P}_{A^{(1)} B^{(1)}} (01) = \oslash$ implies that
$\mathsf{H}_{\textup{inflation}} (010) = \mathsf{H}_{\textup{inflation}} (011) =
\oslash$. After running through all the constraints in~\eqref{Eq:CutNopeRHS},
we get
\begin{equation}
  \vec{\mathsf{H}}_{\textup{inflation}} = (?, \oslash, \oslash, \oslash,
  \oslash, \oslash, \oslash, ?)^{\top} .
\end{equation}
Then, we scan the constraints with right-hand side $= \checked$. For each
constraint, we identify the elements of $\vec{\mathsf{H}}_{\textup{inflation}}$
that correspond to the sum making up the left-hand side. If all these elements
$= \oslash$, we get the contradiction we are looking for. For example
\begin{multline}
\mathsf{P}_{A^{(1)} C^{(1)}} (01) = \mathsf{P}_{A^{(1)} B^{(1)} C^{(1)}}(001) + \\
 + \mathsf{P}_{A^{(1)} B^{(1)} C^{(1)}} (011) = \checked
\end{multline}
whereas
\begin{equation}
\mathsf{H}_{A^{(1)} B^{(1)} C^{(1)}} (001) = \mathsf{H}_{A^{(1)} B^{(1)}
C^{(1)}} (011) = \oslash.
\end{equation}

\subsection{Possibilistic certificates}

Such contradictions give rise to a special type of certificate.

\begin{definition}
  A {\em{possibilistic certificate}} is written:
  \begin{equation}
    \label{Eq:PossibilisticCertificate} T \Rightarrow E_1 \vee E_2 \vee \ldots
    \vee E_m
  \end{equation}
  where each of $T$, $E_1$, \ldots, $E_m$ correspond to a
  {\em{valuation}} of the inflation scenario (see
  Definition~\ref{Def:Valuation}) whose set of observers is AI-expressible.
  The valuation $T$ is the {\em{antecedent}} and the valuations $\{ E_i \}$
  are the {\em{consequents}}. The statement means that if $T$ happens, then
  at least one of the $\{ E_i \}$ must happen.
\end{definition}

We now describe how to extract a certificate from such a contradiction.

\begin{algorithm}
      \label{Alg:ExtractCertificate}When
      Algorithm~\ref{Alg:NSPatInflationBitVector} succeeds, we can extract a
      polynomial inequality that can discriminate probability distributions
      that are close to the one of the pattern, but do not necessarily match a
      realization of the pattern.
\end{algorithm}

We describe the algorithm below. From a possibilistic
certificate~\eqref{Eq:Certificate}, we deduce, using the union bound:
\begin{align}
  P (T) & \leqslant P (E_1 \vee \ldots \vee E_m) \nonumber \\
        & \leqslant P (E_1) + \cdots + P(E_m),
\end{align}
thus, trivially, $\mathsf{P} (T) \leqslant \mathsf{P} (E_1) + \cdots +
\mathsf{P} (E_m)$ at the pattern level. As the antecedent and consequents are
AI-expressible, they all correspond to elements of $\bm{M}$ that can be
computed from the probability distribution of the original scenario.

As an example, we write the following certificate for the cut inflation of the
triangle scenario: its validity is proven later.
\begin{multline}
  \underset{T}{\underbrace{(a^{(1)} = 0, c^{(1)} = 1)}} \quad \Rightarrow \\
  \underset{E_1}{\underbrace{(a^{(1)} = 0, b^{(1)} = 1)}} \vee
  \underset{E_2}{\underbrace{(b^{(1)} = 0, c^{(1)} = 1)}} .
\end{multline}
We get:
\begin{equation}
  \label{Eq:PreCutIneq} P_{A^{(1)} C^{(1)}} (01) \leqslant P_{A^{(1)} B^{(1)}}
  (01)  + P_{B^{(1)} C^{(1)}} (01) .
\end{equation}
Relating the inflation distribution with the original distribution, we have:
\begin{align}
  \label{Eq:CutIneq} 
  P_A (0) P_C (1) \leqslant & P_{A B} (01) + P_{B C} (01), \\
\mathsf{P}_A (0) \mathsf{P}_C (1) \leqslant & \mathsf{P}_{A B} (01) +
  \mathsf{P}_{B C} (01),
\end{align}
which is violated by GHZ-like distributions~\eqref{Eq:PGHZ} both at the
probability and possibilistic level.

How do we find such sets? One way is to post-process the result of the
algorithm presented in Section~\ref{Sec:AlgorithmsNS}.

Every valuation corresponds to a subset of indices in
$\vec{\mathsf{P}}_{\textup{inflation}}$; we write the corresponding sets
$\mathsf{T}$, $\mathsf{E}_1$, \ldots, $\mathsf{E}_M$; when $\mathsf{T}
\subseteq \mathsf{E}_1 \cup \ldots \cup \mathsf{E}_M$ the validity of the
certificate easily follows. As a result of the algorithm, we saw that the
constraint $\mathsf{P}_{A^{(1)} C^{(1)}} (01) = \checked$ leads to a
contradiction; this constraint provides the antecedent $\mathsf{T} = \{ 001,
011 \}$.

Where does this contradiction actually comes from? We present in
Table~\ref{Tab:Investigation} the indices in $\mathsf{T}$ and those of all the
constraints that set elements to $\oslash$. Here, we naturally take
$\mathsf{E}_1 = \{ 010, 011 \}$, $\mathsf{E}_2 = \{ 001, 101 \}$.

\begin{table*}
    $\begin{array}{|l|l|l|l|l|l|l|l|}
      \hline
      a & b & c & \mathsf{P}_{A^{(1)} C^{(1)}} (01) = \checked &
      \mathsf{P}_{A^{(1)} B^{(1)}} (01) = \oslash & \mathsf{P}_{A^{(1)}
      B^{(1)}} (10) = \oslash & \mathsf{P}_{B^{(1)} C^{(1)}} (01) = \oslash &
      \mathsf{P}_{B^{(1)} C^{(1)}} (10) = \oslash\\
      \hline
      0 & 0 & 0 &  &  &  &  & \\
      \hline
      0 & 0 & 1 & \mathsf{T} &  &  & \oslash & \\
      \hline
      0 & 1 & 0 &  & \oslash &  &  & \oslash\\
      \hline
      0 & 1 & 1 & \mathsf{T} & \oslash &  &  & \\
      \hline
      1 & 0 & 0 &  &  & \oslash &  & \\
      \hline
      1 & 0 & 1 &  &  & \oslash & \oslash & \\
      \hline
      1 & 1 & 0 &  &  &  &  & \oslash\\
      \hline
      1 & 1 & 1 &  &  &  &  & \\
      \hline
    \end{array}$
  \caption{\label{Tab:Investigation}Investigation of a contradiction arising
  from the possibilistic algorithm of Section~\ref{Sec:AlgorithmsNS}.}
\end{table*}

To strengthen the resulting inequality, we want to minimize the number of sets
$\mathsf{E}_1, \ldots, \mathsf{E}_M$ that cover $\mathsf{T}$; this corresponds
to the well known set cover problem~{\cite{Korte2012}}.

\section{Local correlations}
\label{Sec:LocalCorrelations}
We now turn our attention to local correlations, which are obtained when the
sources are modeled by local hidden variables. Each source $\mathcal{S}_i$ is
associated with a random variable $s_i \in \Omega_i$. Without loss of
generality~{\cite{Rosset2018}}, we assume the sample space $\Omega_i$ to be
finite, so that $s_i$ has probability distribution $P_{s_i} (s_i)$.

In turn, we describe each observer using a probability distribution
$P_{\mathcal{O}_j | s_{i_{j, 1}} \ldots s_{i_{j, n_j}}} \left( o_j | s_{i_{j,
1}} \ldots s_{i_{j, n_j}} \right)$ where $i \in \{ i_{j, 1}, \ldots, i_{j,
n_j} \}$ if and only if $(\mathcal{S}_i, \mathcal{O}_j) \in \mathcal{E}$.

The resulting correlations are:
\begin{multline}
  \label{Eq:LocalProbabilityDistribution} P (o_1 \ldots o_{| \mathcal{J} |}) =
  \sum_{s_1, \ldots, s_{| \mathcal{I} |}} P_{s_1} (s_1) \cdots P_{s_{|
  \mathcal{I} |}} (s_{| \mathcal{I} |}) \cdot \\
  \cdot \prod_j P_{\mathcal{O}_j |
  s_{i_{j, 1}} \ldots s_{i_{j, n_j}}} \left( o_j | s_{i_{j, 1}} \ldots
  s_{i_{j, n_j}} \right) .
\end{multline}
Without loss of generality, we can assume $P_{s_i} (s_i) > 0$ for all $i$ and
$s_i$; otherwise we can simply reduce the cardinality of the corresponding
sources by removing the values that never happen.

\begin{algorithm}
      \label{Alg:LProbModel}As explicited in~{\cite{Rosset2018}}, the decision
      problem $\vec{P}_{\textup{test}} \in \mathcal{L}$ reduces to a polynomial
      feasibility problem.
\end{algorithm}
    
Consider the triangle scenario. If $\alpha, \beta, \gamma$ are classical
random variables taken from the distributions $P_{\alpha} (\alpha)$,
$P_{\beta} (\beta)$, $P_{\gamma} (\gamma)$, processed by the nodes A,B,C
according to the probability distributions $P_{\textup{A}} (a | \beta \gamma)$,
$P_{\textup{B}} (b | \gamma \alpha)$, $P_{\textup{C}} (c | \gamma \alpha)$, then
the observed outcomes are
\begin{multline}
  \label{Eq:PolyFeasibilityTriangle} P (a b c) = \sum_{\alpha \beta \gamma}
  P_{\alpha} (\alpha) P_{\beta} (\beta) P_{\gamma} (\gamma) \cdot \\
  \cdot P_{\textup{A}} (a |
  \beta \gamma) P_{\textup{B}} (b | \gamma \alpha) P_{\textup{C}} (c | \alpha
  \beta) .
\end{multline}
We can always assume that $\alpha, \beta, \gamma$ are taken from finite sets,
with $| \Omega_{\alpha} |, | \Omega_{\beta} |, | \Omega_{\gamma} | \leqslant
6$, according to~{\cite{Rosset2018}}. The decision problem $\vec{P} \in
\mathcal{L}$ reduces to a feasibility problem including $3 \cdot 5$ degrees of
freedom for the distribution of the classical random variables, and $3 \cdot
36$ degrees of freedom for the nodes processing distributions $P_{\textup{A}},
P_{\textup{B}}, P_{\textup{C}}$. Such nonconvex and highly nonlinear feasibility
problems are usually difficult to handle on their own.

However, looking for a model at the level of patterns is easier.

\begin{algorithm}
  \label{Alg:LPatModel}
  To find whether a pattern $\vec{\mathsf{P}}_{\textup{test}}$ has a model, we collapse the
  model elements to possibilistic distributions. We apply the morphism
  $\varphi$ of Section~\ref{Sec:PossibilitiesAndPatterns}
  to~\eqref{Eq:LocalProbabilityDistribution} and get:

\begin{align}
\label{Eq:PossibilisticLocalSAT}
& \mathsf{P} (o_1 \ldots o_{| \mathcal{J} |}) \nonumber \\
= &  \sum_{s_1, \ldots,
s_{| \mathcal{I} |}} \mathsf{P}_{s_1} (s_1) \cdots \mathsf{P}_{s_{|
 \mathcal{I} |}} (s_{| \mathcal{I} |}) \cdot  \nonumber \\
 & \cdot \prod_j
 \mathsf{P}_{\mathcal{O}_j | s_{i_{j, 1}} \ldots s_{i_{j, n_j}}}
 \left( o_j | s_{i_{j, 1}} \ldots s_{i_{j, n_j}} \right) \nonumber \\
 = &  \sum_{\bm{s}} \prod_j \mathsf{P}_{\mathcal{O}_j |
 s_{i_{j, 1}} \ldots s_{i_{j, n_j}}} \left( o_j | s_{i_{j, 1}} \ldots
 s_{i_{j, n_j}} \right) .
\end{align}
This feasibility problem can be solved using a SAT solver.
\end{algorithm}
                                 
In the triangle scenario, those are $\mathsf{P}_{\alpha}, \mathsf{P}_{\beta},
\mathsf{P}_{\gamma}, \mathsf{P}_{\textup{A}}, \mathsf{P}_{\textup{B}},
\mathsf{P}_{\textup{C}}$. Without loss of generality, we have
$\mathsf{P}_{\alpha} (\alpha) = \mathsf{P}_{\beta} (\beta) =
\mathsf{P}_{\gamma} (\gamma) = \checked$, and such the distribution of the
sources can be omitted from the model. The
equation~\eqref{Eq:PolyFeasibilityTriangle} reduces to:
\begin{equation}
  \label{Eq:PossibilisticLocalEquation} \mathsf{P} (a b c) = \sum_{\alpha
  \beta \gamma} \mathsf{P}_{\textup{A}} (a | \beta \gamma) \mathsf{P}_{\textup{B}}
  (b | \gamma \alpha) \mathsf{P}_{\textup{C}} (c | \alpha \beta)
\end{equation}
We mentioned the inflation technique in the context of the characterization of
the nonsignaling set $\mathcal{N}$, where a complete set of constraints is
obtained by applying rules
\eqref{Eq:CompatibleTypes}~\eqref{Eq:NoDuplication}. Now, keeping
\eqref{Eq:CompatibleTypes}, \eqref{Eq:CompleteSources} but no longer
requiring~\eqref{Eq:NoDuplication}, we use the fact that classical information
can copied. As with Algorithm~\ref{Alg:NSProbInflation}, the resulting
feasibility problem is a linear program. However, there is a systematic way of
constructing a complete hierarchy of such relaxations.

\begin{algorithm}
  \label{Alg:LProbInflation}
  A complete hierarchy of relaxations for the
  decision problem $\vec{P} \in \mathcal{L}$ is obtained by constructing
  the following sequence of inflations~{\cite{Navascues2020}}. For a given
  $n$, we include $n$ copies of each source, and add as many copies of
  observers as possible saturating conditions~\eqref{Eq:CompatibleTypes}
  and \eqref{Eq:CompleteSources}; this is known as the {\em{web
      inflation}}.
\end{algorithm}

\begin{figure*}
  \centering
  \includegraphics{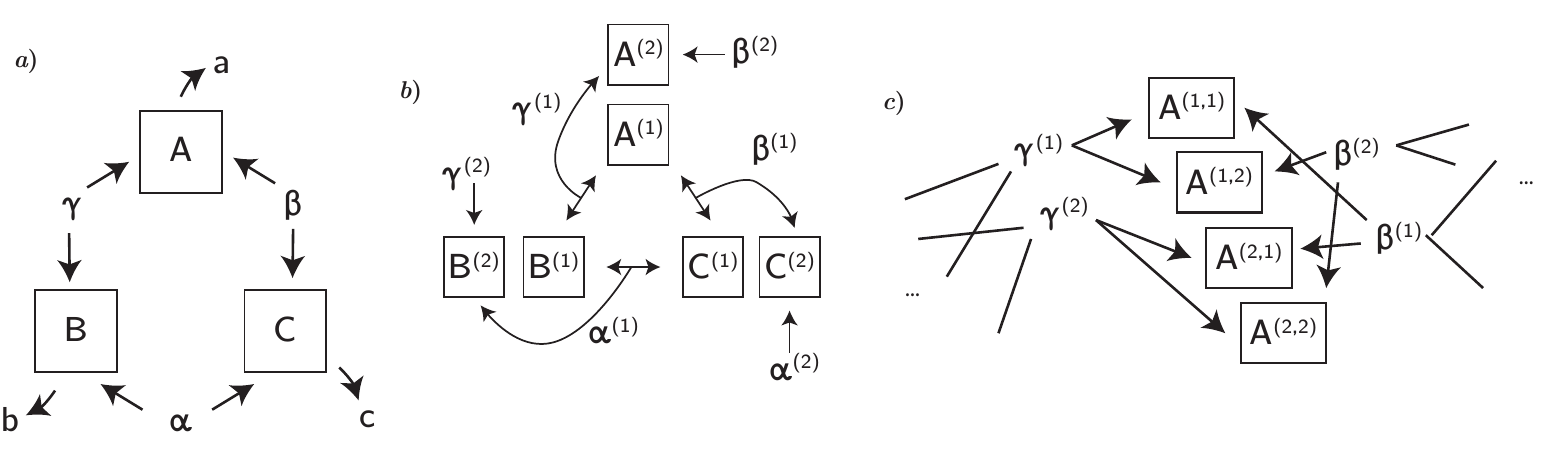}
  \caption{\label{Fig:TriangleLocal}In a), the triangle scenario composed of
  three devices (A, B, C) and three sources ($\alpha, \beta, \gamma$). In b),
  the spiral inflation already mentionned in~{\cite{Wolfe2019}}. In c), part
  of the web inflation for $n = 2$ copies of the sources, where copies of A
  are indexed according to the source copy index.}
\end{figure*}

In Figure~\ref{Fig:TriangleLocal}, we display different inflations for the
triangle scenario. As before, a variant of Algorithm~7 applies to patterns.
\begin{algorithm}
  \label{Alg:LPatInflation}
  The inflation technique applies for patterns as
  well, using either a SAT solver as in
  Algorithm~\ref{Alg:NSPatInflationSat} or filling a bit vector as in
  Algorithm~\ref{Alg:NSPatInflationBitVector}.
\end{algorithm}

There is an additional and often faster algorithm to decide whether a pattern is local or not.

\begin{algorithm}
\label{Alg:LocalTC}
    The possible worlds technique introduced by Fraser in~\cite{Fraser2020} considers the combinatorics of assigning definite valuations to hidden variables, and solves the local set decision problem.
\end{algorithm}

\section{Triangle scenario: results}
\label{Sec:TriangleScenario}
We come back to the triangle scenario of Section~\ref{Sec:TriangleScenarioDef}.
After reduction under symmetry and normalization, we test all 21 orbits. Our
classification is summarized in Figure~\ref{Fig:TriangleClassification}.

\begin{figure}[t]
  \centering
  \includegraphics{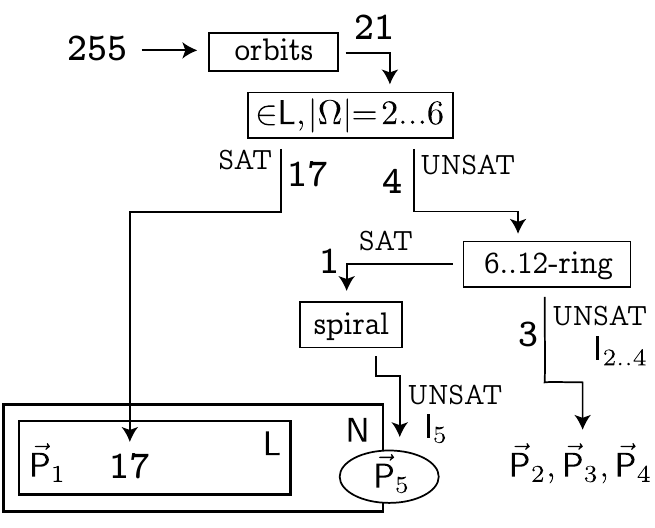}
  \caption{\label{Fig:TriangleClassification}Classification of the
  patterns/orbits in the triangle scenario. For details, see main text.}
\end{figure}

\subsection{Patterns in $\mathsf{L}$}

We apply Algorithm~\ref{Alg:LPatModel} with cardinality 6, which provides a
local model if such a model exists \cite{Rosset2018}. We find that 17 out of 21 orbits have a
local model. We repeat Algorithm~\ref{Alg:LPatModel} with cardinality 2.
Interestingly the results are the same: thus there is no pattern
in~$\mathsf{L}$ that requires a source alphabet of cardinality more than 2,
which is the minimum (cardinality 1 corresponds to a source being omitted).
Representatives of those 17 orbits are given in Appendix~\ref{App:TrianglePatterns}.
These include $\mathsf{P}_{\checked} = [000] + [001] + [010] + [011] + [100] +
[101] + [110] + [111]$ (in all scenarios, the all-$\checked$ pattern is always
local!), the deterministic pattern $\mathsf{P}_{[111]} = [111]$, and the
pattern
\begin{equation}
  \mathsf{P}_1 = [000] + [001] + [010] + [100] .
\end{equation}
The four patterns not in $\mathsf{L}$ are:
\begin{align}
  \mathsf{P}_2 &= [011] + [100], \nonumber \\
  \mathsf{P}_3 &= [011] + [100] + [111], \nonumber \\
  \mathsf{P}_4 &= [011] + [100] + [110] + [111], \nonumber \\
  \mathsf{P}_5 &= [001] + [010] + [100],
\end{align}
noting that $\mathsf{P}_2$ is in the orbit
$(\vec{\mathsf{P}}^{\textup{GHZ}})^G$, cf.~\eqref{Eq:PatGHZOrbit}. The pattern
$\mathsf{P}_5$ corresponds to the pattern of a $W$-type distribution.

\subsection{Patterns not in $\mathsf{N}$}

Using Algorithm~\ref{Alg:NSPatInflationBitVector} and the 6-ring inflation, we
found that $\mathsf{P}_2$, $\mathsf{P}_3$ and $\mathsf{P}_4$ are signaling
enabling. We then increased the relaxation degree to the 9-ring and 12-ring,
and do not find that $\mathsf{P}_5$ is signaling enabling. Is it surprising?
We note that $3 n$-ring inflations only match subgraphs that contain one or
two observers nodes: $\mathsf{P}_{\textup{A}}, \mathsf{P}_{\textup{B}},
\mathsf{P}_{\textup{C}}, \mathsf{P}_{\textup{AB}}, \mathsf{P}_{\textup{AC}},
\mathsf{P}_{\textup{BC}}$. On those marginals, $\mathsf{P}_5$ is
indistinguishable from $\mathsf{P}_1$, and $\mathsf{P}_1 \in \mathsf{N}$
follows from $\mathsf{P}_1 \in \mathsf{L}$.

\subsection{Polynomial inequalities from possibilistic results}

Using Algorithm~\ref{Alg:ExtractCertificate} and the $6$-ring inflation, we
found the following polynomial inequalities, which are respectively violated
by $\mathsf{P}_2$, $\mathsf{P}_3$ and~$\mathsf{P}_4$.

\begin{align}
  I_2 = & P_{\textup{AC}} (11) P_{\textup{A}} (0) + P_{\textup{BC}} (10) P_{\textup{B}}(0)  \nonumber \\
        & - P_{\textup{AB}} (01) P_{\textup{AB}} (10) \geq 0,
\end{align}
\begin{align}
  I_3 = & P_{\textup{AC}} (00) P_{\textup{AC}} (11) + P_{\textup{BC}} (01) P_{\textup{B}} (1) \nonumber \\
      & + P_{\textup{BC}} (10) P_{\textup{BC}} (00) \nonumber \\
      & - P_{\textup{AB}}(01) P_{\textup{AB}} (10) \geq 0,
\end{align}
\begin{align}
  I_4 = & P_{\textup{AC}} (00) P_{\textup{A}} (1) + P_{\textup{BC}} (01) P_{\textup{B}}(1) \nonumber \\
       & - P_{\textup{AB}} (01) P_{\textup{AB}} (10) \geq 0.
\end{align}
The pattern $\mathsf{P}_5$ is detected by the spiral inflation of
Figure~\ref{Fig:TriangleLocal}c, and Algorithm~\ref{Alg:ExtractCertificate}
gives the inequality:
\begin{align}
  I_5 & = P_{\textup{ABC}} (000) \nonumber \\
   & + P_{\textup{AB}} (11) P_{\textup{C}} (1) + P_{\textup{AC}} (11) P_{\textup{B}} (1) \nonumber \\
   &+ P_{\textup{BC}} (11) P_{\textup{A}} (1) \nonumber \\
   &- P_{\textup{A}} (1) P_{\textup{B}} (1) P_{\textup{C}} (1) \geq 0.
\end{align}
However, that certifies that $\mathsf{P}_5 \notin \mathsf{L}$, the status of
this orbit regarding the nonsignaling set is unclear.

\subsection{The case of the $W$-type distribution}
\begin{figure}[t]
  \includegraphics{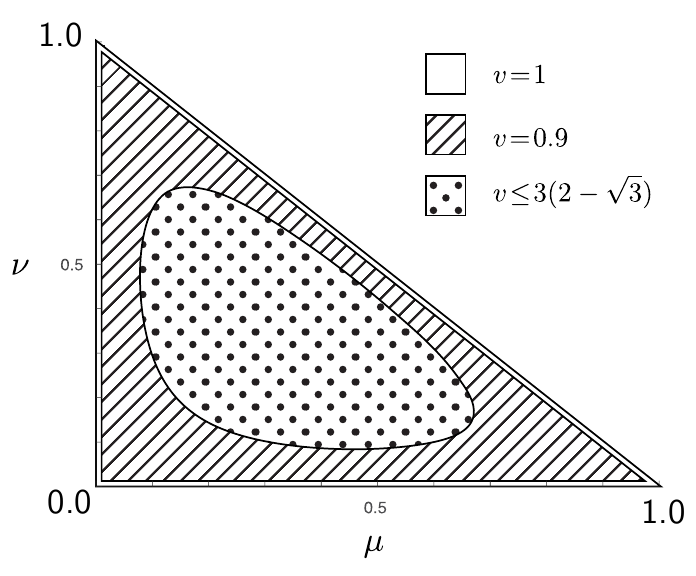}
  \caption{\label{Fig:WVisibility}For $v = 1$, it seems that only a thin
  region around the triangle is feasible for the $6$-ring inflation. As the
  visibility $v$ decreases, this region grows: for $v = 9$, it contains the
  border and the dashed region. When $v \leqslant 3 \left( 2 - \sqrt{3}
  \right) \cong 0.804$, the whole triangle satisfies the $6$-ring inflation.}
\end{figure}

We then test whether {\em{realizations}} $\vec{P}_5^{(a, b)}$ of the pattern
$\mathsf{P}_5$ are in $\mathcal{N}$, using
Algorithm~\ref{Alg:NSProbInflation}:
\begin{equation}
  \vec{P}_5^{(\mu, \nu)} = \mu [001] + \nu [010] + (1 - \mu - \nu) [100],
\end{equation}
where $0 < \mu, \nu, (1 - \mu - \nu) < 1$,
possibly including noise with visibility $v$: 
\begin{equation}
    P_5^{(\mu, \nu, v)} (a b c) = v
P_5^{(\mu, \nu)} (a b c) + (1 - v) / 8.
\end{equation}
In Figure~\todo{\ref{Fig:WVisibility}}, for $v = 1$, it seems at first
sight that most of the considered area is not in $\mathcal{N}$; at $v
\leqslant 3 \left( 2 - \sqrt{3} \right)$, all of it is feasible. By solving
the linear program of Algorithm~\ref{Alg:NSProbInflation} closer and closer to
the border, it becomes clear that the feasible region has a certain thickness
in the neighborhood of the border. While the linear program becomes more and
more numerically instable, we can use exact arithmetic and extract
infeasibility certificates. We plot the results in
Figure~\ref{Fig:WCertificates}.

\begin{figure}[t]
  \includegraphics{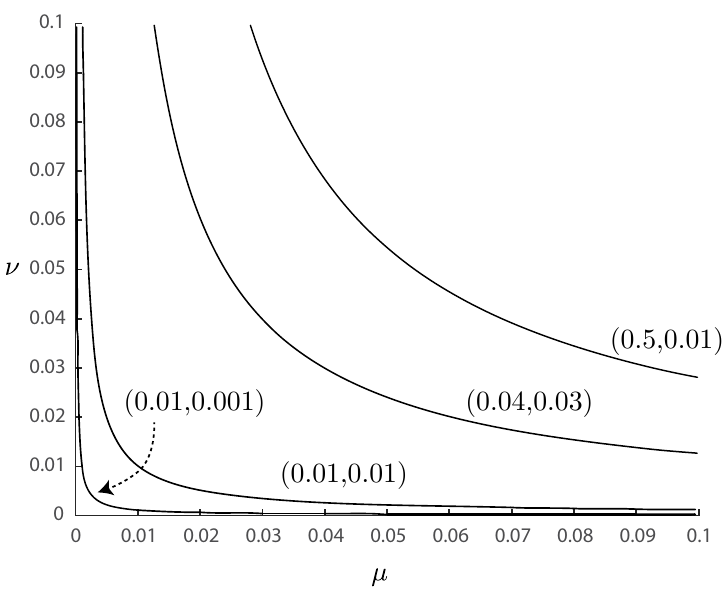}
  \caption{\label{Fig:WCertificates}Certificates found by solving the $6$-ring
  inflation LP for $P_5^{(\mu_0, \nu_0)}$ according to four pairs $(\mu_0,
  \nu_0)$, shown along the certificates.}
\end{figure}

While we could not rule out the pattern $\mathsf{P}_5$ using the possibilistic
semiring, it may well be that all its realizations are ruled out by using
Algorithm~\ref{Alg:NSProbInflation} on a suitable $3 n$-ring inflation. We
leave that question as an open problem.

\section{Square scenario: results}
\label{Sec:SquareScenario}
We address the square scenario with binary outcomes, as shown in Figure~\ref{Fig:Scenarios} and in Figure~\ref{Fig:SquareInflations}a.
Our classification process is summarized in Figure~\ref{Fig:SquareClassification}.
First, the patterns group into 804 orbits under symmetry.
We then verify if the pattern obeys the factorization condition:
\begin{equation}
\label{Eq:SquareFactorization}
\mathsf{P}_\textup{AC}(ac) = \mathsf{P}_\textup{A}(a) \mathsf{P}_\textup{C}(c), 
\mathsf{P}_\textup{BD}(bd) = \mathsf{P}_\textup{B}(b) \mathsf{P}_\textup{D}(d), 
\end{equation}
which should always be satisfied: $285$ orbits fail this test and can be pruned.
There remains $519$ interesting patterns to classify.

\subsection{Patterns in $\mathsf{L}$}

Applying the results of~{\cite{Rosset2018}}, all local models use classical
random variables with an alphabet size of at most 12. Due to computational
constraints, we could only run Algorithm~\ref{Alg:LPatModel} with alphabet
size up to $6$. For 95 orbits, we found a model using
alphabet size 2. For 21 additional orbits, we found a model using alphabet
size 3. Increasing the alphabet size from $3$ to $6$ did not help further.

For the remaining 403 orbits, we cannot rule out the possibility of a model
for an alphabet size from 7 to 12. We turn to the possibilistic inflation technique.

\subsection{Patterns not in $\mathsf{N}$}

We applied Algorithm~\ref{Alg:NSPatInflationSat} on those $403$ orbits
using two different nonsignaling inflations: the 8-ring and the 12-ring, corresponding to the
inflations in Figure~\ref{Fig:SquareInflations}b and Figure~\ref{Fig:SquareInflations}c respectively.
The 8-ring finds $329$ orbits to be signaling-enabling, while the 12-ring finds $19$
additional orbits that are also signaling-enabling.
We are left with $55$ orbits whose status is yet unknown. We then apply
the local 2-web inflation, as described in Figure~\ref{Fig:TriangleLocal}c and Figure~\ref{Fig:SquareInflations}d, with two
copies of each source. 
We have 49 patterns that failed to satisfy the 2-web inflation: for sure,
those patterns are $\notin \mathsf{L}$, but their membership status regarding $\mathsf{N}$
is unknown.
Remains 6 patterns that satisfy the 2-web inflation. At that point, we turn to the possible worlds technique (Algorithm~\ref{Alg:LocalTC}) and find that those 6 patterns are $\notin \mathsf{L}$.
In particular, this means that all the patterns in $\mathsf{L}$ can be realized with models of LHV cardinality 3.

\begin{figure}[t]
  \centering
  \includegraphics{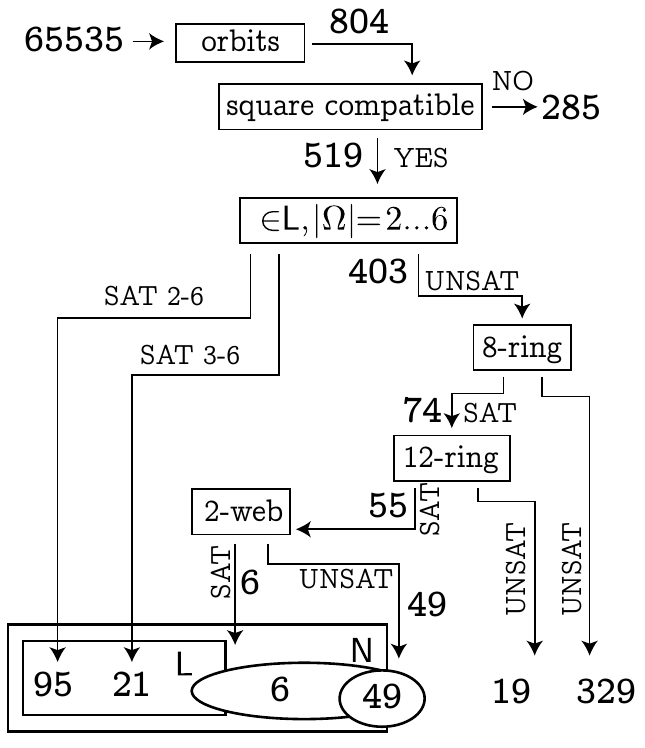}
  \caption{
  \label{Fig:SquareClassification}
  Classification of patterns/orbits in the square scenario, see main text for the process.
  }
\end{figure}

\begin{figure*}
  \centering
  \includegraphics{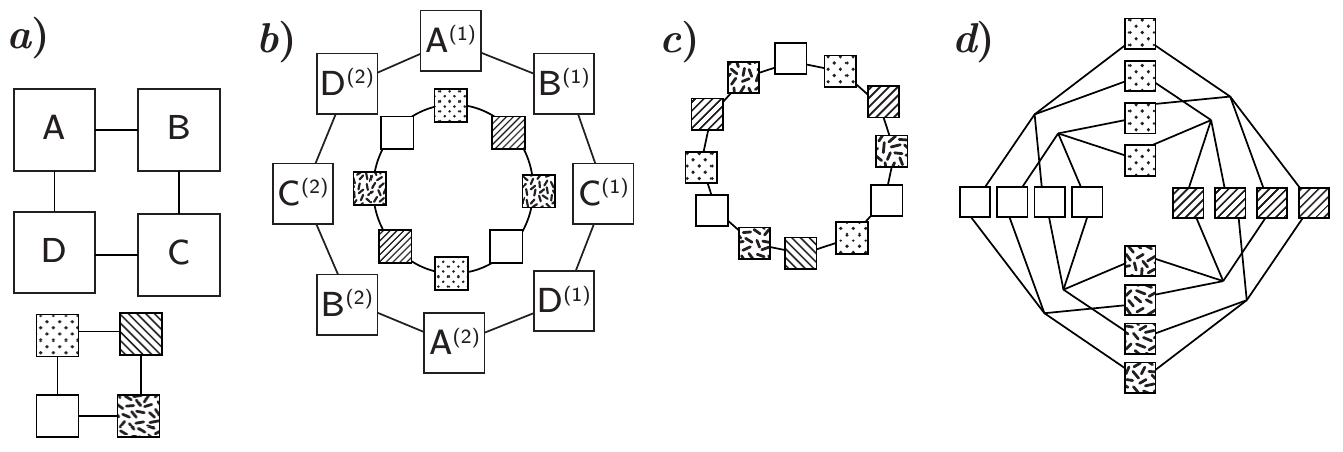}
  \caption{\label{Fig:SquareInflations}
  a) Square scenario, including simplified drawing. b) 8-ring inflation, including simplified drawing. 
  c) Simplified drawing of the 12-ring inflation. d) Simplified drawing of the 2-web inflation.}
  
\end{figure*}

\subsection{Quantum nonlocality in the square network}

Interestingly, the characterisation of patterns that are nonlocal, i.e. not in $\mathsf{L}$, allows us to provide a simple example of a quantum nonlocal distribution in the square network. Of particular interest is the fact the low cardinality of the outcomes, which are all binary (while there are no inputs). Note that it is still an open question whether nonlocality is possible in the triangle network with binary outputs and no inputs, see e.g. \cite{Boreiri2022}.

The construction is based on the well-known Hardy paradox \cite{Hardy}, a demonstration 
of standard quantum Bell nonlocality based on a possibilitic argument. We start with a standard bipartite Bell test. The first party (Alice) receives a binary input $x$ (from another party called Dave) and provides an binary output $a$. Similarly, the second party (Bob) receives a binary input $y$ (from Charlie) and outputs a bit $b$. One usually considers the conditional joint probability $P(ab|xy)$. In Hardy's construction, Alice and Bob perform well-chosen local Pauli measurements on a non-maximally entangled state of two-qubits, see \cite{Hardy} for details. The key point is that the resulting probability distribution has the following feature: 

\begin{equation}
    P(10|10) = P(01|01) = P(00|11)=0 \,
\end{equation}
while all other probabilities are strictly positive, including $P(00|00)$, which leads to a logical contradiction for any local model. 

Now we embed this bipartite Bell test in the square network, adapting the idea of Fritz \cite{Fritz_2012}. We place Alice, Bob, Charlie and Dave in a square network. The source connecting Alice and Bob distributes the required two-qubit entangled state. The source between Alice and Dave distributes the inputs $x$, while the source between Bob and Charlie distributes the input $y$. Note that the source between Charlie and Dave is here useless. Alice and Bob output $a$ and $b$, respectively. Charlie and Dave output $y$ and $x$, respectively. We obtain a quantum distribution on the square with binary outputs. Using the SAT algorithm, it can be checked that the pattern of this quantum distribution is not in $\mathsf{L}$. 

This shows that the square network with binary outputs allows for quantum nonlocality. We note that a similar construction shows that when Alice and Bob share a Popescu-Rohrlich nonlocal box \cite{PR}, the resulting pattern is not in $\mathsf{L}$. This suggests that the square network with binary output could also feature stronger than quantum nonlocality.

\section{Relaxation of the independence assumption}
\label{Sec:Relaxation}
The study of nonlocality in networks assumes that sources are perfectly
independent
\begin{align}
  P_{\bm{s}} \left( \bm{s} \right) = & P_{s_1 \ldots s_{|
  \mathcal{I} |}} (s_1 \ldots s_{| \mathcal{I} |})  \\
  = & P_{s_1} (s_1) \ldots
  P_{s_{| \mathcal{I} |}} (s_{| \mathcal{I} |}) = \prod_{i \in \mathcal{I}} P_{s_i} (s_i) \nonumber
\end{align}
whereas recent works have challenged this assumption by considering that the
sources in the network could be classically correlated~{\cite{Supic2020}} (see also \cite{Chaves2021}). In
this section, we extend inequalities obtained through
Algorithm~\ref{Alg:ExtractCertificate} to be still valid when the independence
assumption is relaxed. We model the correlations among the $\{ s_i \}$ by
introducing another random variable $\lambda = 1, \ldots, \Lambda$ distributed
according to $P (\lambda)$, as in Figure~\ref{Fig:Relaxing}. %Still we require%
Since the original sources in the network may be affected by the latter, we need to bound its causal influence to avoid that any distribution $P_{\mathcal{O}|\mathbf{s}}$ can be reproduced via a local model by setting ($s_1=\cdots=s_{|\mathcal{I}|} = \lambda$) with $\lambda$ sampled from $P_{\mathcal{O}|\mathbf{s}}$. We therefore impose the condition
\begin{equation}
  \label{Eq:LowerBound} \epsilon_1 \prod_{i \in \mathcal{I}} P_{s_i} (s_i)
  \quad \leq \quad P_{\bm{s} | \lambda} \left( \bm{s} | \lambda
  \right)
\end{equation}
for some $\epsilon_1 \in (0, 1]$. Without loss of generality, we require $P
(\lambda) > 0$. We write $\mathcal{L}_{\epsilon_1}$ the set of probability
distributions satisfying a relaxed condition similar
to~\eqref{Eq:LocalProbabilityDistribution}:
\begin{equation}
  \label{Eq:ObservedProbabilitiesRelaxation} 
  \begin{split}
  P (o_1 \ldots o_{| \mathcal{J}
  |}) = \sum_{\lambda} P (\lambda) {\sum_{\bm{s}}}  P_{\bm{s} |
  \lambda} \left( \bm{s} | \lambda \right) \cdot 
  \\ 
  \prod_j
  P_{\mathcal{O}_j | s_{i_{j, 1}} \ldots s_{i_{j, n_j}}} \left( o_j | s_{i_{j,
  1}} \ldots s_{i_{j, n_j}} \right) .
  \end{split}
\end{equation}
Similarly, we write $\mathsf{L}_{\epsilon_1}$ the set of patterns achievable
with a given $\epsilon_1$. For $\epsilon_1 = 1$, we get back
$\mathcal{L}_{\epsilon_1} =\mathcal{L}$, while for $\epsilon_1 = 0$, any
normalized $\vec{P}$ can be reached by embedding the desired outcomes in
$\lambda = (o_1, \ldots, o_{| \mathcal{I} |})$ and passing the relevant values
through the sources.

We make an interesting observation in the intermediate regime.

\begin{proposition}
  For $0 < \epsilon_1 < 1$, we have $\mathsf{L}_{\epsilon_1} = \mathsf{L}$.
  Thus, if a pattern is not local, it is not local even in the presence in
  limited correlation between the sources.
\end{proposition}

\begin{proof}
  We apply the semiring morphism $\varphi : \mathbb{R}^{\geqslant 0}
  \rightarrow \mathbb{P}$ of Section~\ref{Sec:PossibilitiesAndPatterns} to
  Eq.~\eqref{Eq:ObservedProbabilitiesRelaxation} to obtain:
  \begin{align}
      & \mathsf{P} (o_1 \ldots o_{| \mathcal{J} |}) = \nonumber \\
      = & \sum_{\lambda}
      \mathsf{P} (\lambda) {\sum_{\bm{s}}}  \mathsf{P}_{\bm{s} |
      \lambda} \left( \bm{s} | \lambda \right) \cdot \nonumber \\
      & \cdot \prod_j
      \mathsf{P}_{\mathcal{O}_j | s_{i_{j, 1}} \ldots s_{i_{j, n_j}}} \left(
      o_j | s_{i_{j, 1}} \ldots s_{i_{j, n_j}} \right) \nonumber \\
       = & \sum_{\lambda} \sum_{\bm{s}} \prod_j
      \mathsf{P}_{\mathcal{O}_j | s_{i_{j, 1}} \ldots s_{i_{j, n_j}}} \left(
      o_j | s_{i_{j, 1}} \ldots s_{i_{j, n_j}} \right) \nonumber \\
       = & \sum_{\bm{s}} \prod_j \mathsf{P}_{\mathcal{O}_j | s_{i_{j,
      1}} \ldots s_{i_{j, n_j}}} \left( o_j | s_{i_{j, 1}} \ldots s_{i_{j,
      n_j}} \right)
    \end{align}
  where we made use of $\mathsf{P}_{\bm{s}} \left( \bm{s} \right)
  = \checked$ for all $\bm{s}$ by assumption, and \ $\epsilon_1 \prod_{i
  \in \mathcal{I}} P_{s_i} (s_i) \quad \leq \quad P_{\bm{s} | \lambda}
  \left( \bm{s} | \lambda \right)$ implies $P_{\bm{s} | \lambda}
  \left( \bm{s} | \lambda \right) = \checked$ for all $\bm{s},
  \lambda$. The final equation is identical
  to~Eq.~\eqref{Eq:PossibilisticLocalSAT}, and thus the feasible sets are the
  same.
\end{proof}

In essence, the condition $\epsilon_1 > 0$ limits the strength of the ``common
cause'' $\lambda$ correlating the sources so that it cannot veto a particular
possible value assignment.

\subsection{Inequalities with relaxation of the independence condition}

Due to the noise present in experimental observations, it is desirable to
derive inequalities that can certify that $\vec{P} \notin \mathcal{L}$ when
the observed $\vec{P}$ is close but does not conform to a particular pattern,
in the presence of some amount of correlation between sources.

As a consequence of (\ref{Eq:LowerBound}) we derive the following upper bound on the correlations between the sources, preventing from deterministic assignments of the sources $\{s_i\}$,
\begin{equation}
  \label{Eq:UpperBound} P_{\bm{s} | \lambda} \left( \bm{s} |
  \lambda \right) \quad \leq \quad \epsilon_2 \prod_{i \in \mathcal{I}}
  P_{s_i} (s_i)
\end{equation}
for some $\epsilon_2 > 1$, which can be deduced from the value of $\epsilon_1
> 0$. Note that $\epsilon_2 = 1$ would lead to $\mathcal{L}_{\epsilon_1,
\epsilon_2} =\mathcal{L}$.

We consider probabilistic inequalities derived from possibilistic certificates
of the form~\eqref{Eq:PossibilisticCertificate}, obtained first without
relaxing the independence condition. They are written according to valuations
of the inflation scenario:
\begin{align}
  P (T) & \leqslant P (E_1 \vee \ldots \vee E_m) \nonumber \\
  & \leqslant P (E_1) + \cdots + P
  (E_m) .
\end{align}
All of these valuations are AI-expressible, meaning that they correspond to
monomials in $\bm{M}$, which we write $M_T, M_{E_1}, \ldots, M_{E_m}$.
We write the monomial value evaluated on a test distribution $\vec{P}$ as $M_T
(\vec{P}),$ etc.

The original inequality assuming independence of sources is written:
\begin{equation}
  M_T (\vec{P}) \leqslant M_{E_1} (\vec{P}) + \cdots + M_{E_m} (\vec{P}) .
\end{equation}
Consider $P (V)$ where $V$ is a valuation AI-expressible in the inflation
scenario. It corresponds to a monomial $M_V \in \bm{M}$ that can be
computed from the probability distribution of the original scenario. Abusing
the notation, we write $\deg V = \deg M_V$ the number of connected components
in the subgraph containing as observers only the one occuring in the
valuation.

Now, sources in the inflation scenario can be correlated through $\lambda$, as
shown in Figure~\ref{Fig:Relaxing}.

\begin{lemma}
  Let $V$ be a AI-expressible valuation. Then, there exists a $n$ such that:
  \begin{equation}
    \epsilon^n_1 M_V (\vec{P}_{\textup{original}}) \leqslant P (V) \leqslant
    \epsilon_2^n M_V (\vec{P}_{\textup{original}})
  \end{equation}
\end{lemma}

\begin{proof}
  Let $V_1, \ldots, V_N$ be the valuations of the connected components in the
  inflation subgraph addressed by $V$, and $M_{V_1}, \ldots, M_{V_N} \in
  \bm{M}$ the corresponding monomials in the coefficients of the
  original probability distribution. We have
  \begin{equation}
\begin{split}
    P (V) = \sum_{\lambda, \bm{s}} P (\lambda) P \left(
    \bm{s}^{(1)} | \lambda \right) \ldots P \left( \bm{s}^{(M)} |
    \lambda \right) \cdot \\ 
    P \left( V_1 | \bm{s}_{V_1} \right) \ldots P \left(
    V_N | \bm{s}_{V_N} \right),
    \end{split}
  \end{equation}
  where $\bm{s}^{(1)}, \ldots, \bm{s}^{(n)}$, for a given $n$, are
  the $n$ copies of sources present in the inflation (without loss of
  generality, we allow source copies not connected to any observer), and
  $\bm{s}_{V_{\ell}}$ regroups the sources needed to compute the outcome
  values present in the valuation component $V_{\ell}$.
  
  The connected components $V_1, \ldots, V_N$ are connected to distinct
  sources, and contain at most one copy of each original source, as all the
  components are AI-expressible. Without modifying $P (V)$, we relabel sources
  so that $\bm{s}_{V_1}$ contains only sources of the form
  $\mathcal{S}_i^{(1)}$, $\bm{s}_{V_2}$ contains only sources of
  the form $\mathcal{S}_i^{(2)}$ and so on. We have:
  \begin{align}
    P (V) = & \sum_{\lambda} P (\lambda) \left[ \sum_{\bm{s}_{V_1}} P
    \left( \bm{s}_{V_1} | \lambda \right) P \left( V_1 |
    \bm{s}_{V_1} \right) \right] \ldots \nonumber \\
    & \ldots \left[ \sum_{\bm{s}_{V_N}}
    P \left( \bm{s}_{V_N} | \lambda \right) P \left( V_N |
    \bm{s}_{V_N} \right) \right] .
  \end{align}
  We now apply~\eqref{Eq:LowerBound} and~\eqref{Eq:UpperBound} to each term
  $\sum_{\bm{s}_{V_{\ell}}} P \left( \bm{s}_{V_{\ell}} | \lambda
  \right) P \left( V_{\ell} | \bm{s}_{V_1} \right)$ and obtain:
  \begin{align}
    \epsilon_1 M_{\ell} (\vec{P}_{\textup{original}}) & \leqslant
    \sum_{\bm{s}_{V_{\ell}}} P \left( \bm{s}_{V_{\ell}} | \lambda
    \right) P \left( V_{\ell} | \bm{s}_{V_{\ell}} \right) \nonumber \\
    & \leqslant \epsilon_2 M_{\ell} (\vec{P}_{\textup{original}})
  \end{align}
  and thus:
  \begin{equation}
    \epsilon_1^n M_V (\vec{P}_{\textup{original}}) \leqslant P (V) \leqslant
    \epsilon_2^n M_V (\vec{P}_{\textup{original}}) .
  \end{equation}
  
\end{proof}

Using this last lemma, we are ready to write an inequality valid in presence
of correlations between sources:
\begin{align}
  \epsilon_1^n M_T (\vec{P}) & \leqslant P (T) \\
   & \leqslant P (E_1) + \cdots + P(E_m) \nonumber \\
   & \leqslant \epsilon_2^n [M_{E_1} (\vec{P}) + \ldots + M_{E_m}
  (\vec{P})] . \nonumber
\end{align}
As an example, the inequality~\eqref{Eq:CutIneq} comes from a certificate of
the form~\eqref{Eq:Certificate} which is interpreted probabistically
as~\eqref{Eq:PreCutIneq}:
\begin{equation}
  P_{A^{(1)} C^{(1)}} (01) \leqslant P_{A^{(1)} B^{(1)}} (01)  + P_{B^{(1)}
  C^{(1)}} (01),
\end{equation}
and translates to
\begin{equation}
  (\epsilon_1)^2 P_A (0) P_C (1) \leqslant \epsilon_2 [P_{A B} (01) + P_{B C}
  (01)] .
\end{equation}
\begin{figure}[t]
  \centering
  \includegraphics{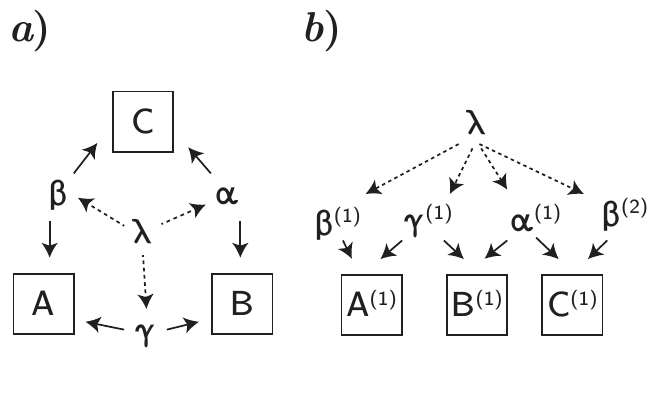}
  \caption{\label{Fig:Relaxing}Relaxing the independence assumption in the
  triangle scenario}
\end{figure}

\section{Conclusion}
We explored the possibilistic approach to network locality. 
We derived efficient algorithms to test whether a given possibilistic pattern is compatible with local (classical) or nonsignaling resources. We also considered the action of the
symmetries of the scenario on those patterns to group them into orbits as to simplify the problem and hence reduce its
computational costs.

In turn, we applied these methods to the triangle and square networks. In the case of the triangle network with binary outcomes, most of the orbits are local, and
we could find explicit local models using binary classical variables. We found that
three orbits were signaling-enabling (hence incompatible with the triangle network), while the orbit corresponding to the so-called ``W pattern'' exhibits
an interesting feature. It cannot be revealed as signaling-enabling, however, it seems that
the vast majority of its realizations are indeed signaling-enabling at a low ring inflation level.
This is our first open question: can any realization of the W pattern be detected (as signaling enabling) by a ring inflation?

In the case of the square network with binary outcomes, the majority of patterns is signaling enabling; some
of the patterns do not even respect the independence relation between observers. We listed the $116$ patterns that are compatible with a local model, and found that all of them can be reached with local variables of cardinality 3 at most. Additionally, there are $55$ patterns which are provably nonlocal, yet we do not know if they are non-signaling.
This ambiguity is a fundamental limit of the current algorithms: we can find whether $\vec{\mathsf{P}} \notin \mathsf{N}$,
but we do not have an algorithm to construct models of $\vec{\mathsf{P}} \in \mathsf{N}$. This is our second open question. 

As an application of our methods, we provided an example of quantum nonlocality in the square network with binary outcomes. This represents an instance of a quantum nonlocal distribution with low output cardinality.
Are there nonlocal patterns in other orbits that have nevertheless a quantum realization?

Finally, we described how the inequalities recovered from possibilistic certificates are robust against correlations between
the sources, relaxing the independence condition typically used in network locality. It would be interesting to see if this 
approach can be generalized to correlations between non-signaling sources.

As future avenues for research, there are natural scenarios to consider next under the possibilistic angle: other
topologies, going beyond binary outcomes, adding inputs (different measurement settings) to the observers. However,
even with the reduction coming from the grouping into orbits, the number of orbits will be pretty large.
In the triangle scenario, the most interesting patterns. such as the W pattern, were also the patterns symmetrical under
permutation of parties. Would it be then interesting in those scenarios 
to study a subset of orbits obeying a particular symmetry?

\medskip

\emph{Acknowledgements.---} We thank Elie Wolfe for bringing to our attention the possible worlds technique of Ref. \cite{Fraser2020} and sharing his code implementation, as well as for comments on a first version of this manuscript. We also thank Jean-Daniel Bancal, Marie Ioannou and Eloïc Vallée for discussions. We acknowledge financial support from the Swiss National Science Foundation (project 2000021\_192244/1 and NCCR SwissMAP) and from the European Union’s Horizon 2020 research and innovation program under the Marie Skłodowska-Curie grant agreement No 956071.

\appendix
\section{Triangle patterns generation and classification}
\label{App:TrianglePatterns}
Source code available at \url{https://github.com/Possibilistic-network-nonlocality/main/tree/main/triangle%20scenario}
\section{Square patterns generation and classification}
\label{App:SquarePatterns}
Source code available at \url{https://github.com/Possibilistic-network-nonlocality/main/tree/main/square%20scenario}

\bibliographystyle{quantum}
\bibliography{biblio}

\end{document}